\documentclass[11pt,a4paper]{amsart}
\usepackage{fullpage}

\usepackage{amsmath,amssymb,amsthm} \usepackage{mathtools}
\usepackage[bookmarks=true,hypertexnames=false,pagebackref]{hyperref} \hypersetup{colorlinks=true,
  citecolor=blue, linkcolor=red, urlcolor=purple}

\usepackage[textsize=tiny]{todonotes} 
 \usepackage{algorithm}
\usepackage{algpseudocode}

\usepackage{xifthen}

\newcommand{\Cone}{357}
\newcommand{\Ctwo}{931}
\newcommand{\kone}{28}

\newcommand{\Cthreshold}{\ensuremath{\Cone \Delta^{\frac{14}{k-14}}}}
\newcommand{\Sthreshold}{\ensuremath{\Ctwo \Delta^{\frac{16}{k-16/3}}}}

 \newcommand{\abs}[1]{\left\vert#1\right\vert}
\newcommand{\set}[1]{\left\{#1\right\}} \newcommand{\tuple}[1]{{\left(#1\right)}}
\newcommand{\eps}{\varepsilon}  \newcommand{\tp}{\tuple}
\renewcommand{\mid}{\;\middle\vert\;} \newcommand{\cmid}{\,:\,} 
 \newcommand{\defeq}{\triangleq} 
\newcommand{\ul}{\underline} \def\*#1{\mathbf{#1}} \def\+#1{\mathcal{#1}} \def\-#1{\mathrm{#1}}
 \newcommand{\poly}{\mathrm{poly}} \newcommand{\ceil}[1]{\left\lceil #1
  \right\rceil} \newcommand{\floor}[1]{\left\lfloor #1 \right\rfloor}

\renewcommand{\defeq}{:=}

\renewcommand{\Pr}[2][]{ \ifthenelse{\isempty{#1}} {\mathbf{Pr}\left[#2\right]}
  {\mathbf{Pr}_{#1}\left[#2\right]} } \newcommand{\mPr}[2][]{ \ifthenelse{\isempty{#1}}
  {\mathop{\mathbf{Pr}}\left[#2\right]} {\mathop{\mathbf{Pr}}_{#1}\left[#2\right]} }

\newcommand{\E}[2][]{ \ifthenelse{\isempty{#1}} {\mathop{\mathbf{E}}\left[#2\right]}
  {\mathop{\mathbf{E}}_{#1}\left[#2\right]} }

\newtheorem{theorem}{Theorem} \newtheorem{lemma}[theorem]{Lemma}
\newtheorem{corollary}[theorem]{Corollary} \newtheorem{definition}[theorem]{Definition}
 
\newtheorem{claim}[theorem]{Claim}  

\newcommand\blfootnote[1]{%
  \begingroup \renewcommand\thefootnote{}\footnote{#1}%
  \addtocounter{footnote}{-1}%
  \endgroup }

\begin{document}
\title{Counting Hypergraph Colourings in the Local Lemma
  Regime}
\blfootnote{ A preliminary version of
    this paper appeared in 50th Annual ACM SIGACT Symposium on the
    Theory of Computing (STOC), 2018, Los Angeles.}


\author{Heng Guo} \address[Heng Guo]{School of Informatics, University of Edinburgh, Informatics
  Forum, Edinburgh, EH8 9AB, United Kingdom.}  \email{hguo@inf.ed.ac.uk}


\author{Chao Liao} \address[Chao Liao]{Department of Computer Science and Engineering, Shanghai Jiao
  Tong University, 800 Dongchuan Road, Minhang District, Shanghai, China.}
\email{chao.liao.95@gmail.com}


\author{Pinyan Lu} \address[Pinyan Lu]{ITCS, Shanghai University of Finance and Economics, 100
Wudong Road, Yangpu District, Shanghai, China.}  \email{lu.pinyan@mail.shufe.edu.cn}


\author{Chihao Zhang} \address[Chihao Zhang]{John Hopcroft Center for
  Computer Science, Shanghai Jiao Tong
  University, 800 Dongchuan Road, Minhang District, Shanghai, China.}  \email{chihao@sjtu.edu.cn}

\begin{abstract}
  We give a fully polynomial-time approximation scheme (FPTAS) to count the number of $q$-colourings
  for $k$-uniform hypergraphs with maximum degree $\Delta$ if $k\ge \kone$ and $q > \Cthreshold$ .
  We also obtain a polynomial-time almost uniform sampler if $q>\Sthreshold$.  These are the first
  approximate counting and sampling algorithms in the regime $q\ll\Delta$ (for large $\Delta$ and
  $k$) without any additional assumptions.  Our method is based on the recent work of Moitra (STOC,
  2017).  One important contribution of ours is to remove the dependency of $k$ and $\Delta$ in
  Moitra's approach.
\end{abstract}

\maketitle

\section{Introduction}

Hypergraph colouring is a classic and important topic in combinatorics.  Its study was initiated by
Erd\H{o}s' seminal result \cite{Erd63}, a sufficient upper bound on the number of edges so that a
uniform hypergraph is $2$-colourable.  Many important tools in the probabilistic method have been
developed around this subject, such as the Lov\'asz local lemma \cite{EL75}, and the R\"odl nibble
\cite{Rod85}.

In this paper, we consider the problem of approximately counting colourings in $k$-uniform
hypergraphs.  The most successful approach to approximate counting is Markov chain Monte Carlo
(MCMC).  See \cite{DFK91,JS93,JSV04} for a few famous examples.  Indeed, MCMC has been extensively
studied for graph colourings in low-degree graphs.  Jerrum \cite{Jer95} showed that the simple and
natural Markov chain, Glauber dynamics, mixes rapidly, if $q>2\Delta$, where $q$ is the number of
colours and $\Delta$ is the maximum degree of the graph.  As a consequence, there is a \emph{fully
  polynomial-time randomized approximation scheme} (FPRAS) for the number of colourings if
$q>2\Delta$.  This result initiated a series of research and the best
bound in general requires that $q>(11/6-\eps) \Delta$ for some small
constant $\eps>0$ \cite{Vig00,CDMPP19}.  It is conjectured that Glauber dynamics is
rapidly mixing if $q>\Delta+1$, the ``freezing'' threshold, but current evidences typically require
extra conditions in addition to the maximum degree \cite{HV03,DFHV13}.  On the flip side, see
\cite{GSV15} for some (almost tight) NP-hardness results.

In $k$-uniform hypergraphs, the Markov chain approach still works, if $q > C\Delta$ for $C=1$ when
$k\ge4$ and $C = 1.5$ when $k = 3$ \cite{BDK08,BDK06}.  However, the local lemma implies that a
hypergraph is $q$-colourable if $q > C\Delta^{1/(k-1)}$ for some constant $C$.  This threshold is
much smaller than $\Delta$ when $\Delta$ is large.  Moser and Tardos' algorithmic version of the
local lemma \cite{MT10} implies that we can efficiently find a $q$-colouring under the same
condition.  Indeed, the study of the algorithmic local lemma has been a highly active area.  See
\cite{KS11,HSS11,HS13b,HS13a,HV15,AI16,Kol16,CPS17,HLLWX17} for various recent development.

In view of the success of algorithmic local lemma, it is natural to wonder, whether we can also
randomly generate hypergraph colourings, or equivalently, approximately count their number, beyond
the $q\asymp\Delta$ bound and approaching $q\asymp\Delta^{1/(k-1)}$?  Unfortunately, designing
Markov chains quickly runs into trouble if $q\ll\Delta$.  ``Freezing'' becomes possible in this
regime (see \cite{FM11} for examples\footnote{Interestingly, to prove the existence of frozen
  colourings, we also need to appeal to the local lemma.}), and the state space of proper hypergraph
colourings may not be connected via changing the colour of a single vertex, the building block move
of Glauber dynamics.

The only successful application of MCMC in this regime is due to Frieze et al.~\cite{FM11,FA17},
which requires that $q>\max\{C_k\log n$, $500 k^3\Delta^{1/(k-1)}\}$ and the hypergraph is
simple.\footnote{A hypergraph is simple if the intersection of any two hyperedges contains at most
  one vertex.}  Here $q=\Omega(\log n)$ is necessary to guarantee that ``frozen'' colourings are not
prevalent.  Furthermore, it is reasonable to believe that simple hypergraphs are much easier
algorithmically than general ones, since their chromatic numbers are
$O\left(\frac{\Delta}{\log\Delta}\right)^{1/(k-1)}$ \cite{FM13}, significantly smaller than the
bound implied by the local lemma, and related Glauber dynamics for hypergraph independent sets works
significantly better in simple hypergraphs than in general ones \cite{HSZ19}.

Our main result is a positive step beyond the freezing barrier in general $k$-uniform hypergraphs.
Our result also answers some open problems raised in \cite{FM11}.

\begin{theorem} \label{thm:main} For integers $\Delta\ge 2$, $k\ge \kone$, and $q>\Cthreshold$,
  there is an FPTAS for $q$-colourings in $k$-uniform hypergraphs with maximum degree $\Delta$.
\end{theorem}

When $k$ and $\Delta$ are large, our result is better than the Markov chain results
\cite{BDK08,BDK06} and gets into the freezing regime.  The exponent of our polynomial time bound
depends on the constants $k$ and $\Delta$.

Our method is based on an intriguing result shown by Moitra \cite{Moi19} recently, who gave
\emph{fully polynomial-time deterministic approximation schemes} (FPTAS) to count satisfying
assignments of $k$-CNF formulas in the local lemma regime.  It is not hard to see that Moitra's
approach is rather general, and indeed it works for hypergraph colourings if some strong form of the
local lemma condition holds, and $k \ge C\log\Delta$ for some constant $C$, without any requirement
on the connectedness of the state space.  Unfortunately, the requirement that $k\ge C\log\Delta$ is
necessary for a ``marking'' argument to work in Moitra's approach.  This is not an issue for $k$-CNF
formulas, as in that setting the (strong) local lemma condition dictates that $k\ge C\log\Delta$.
However, for hypergraph colourings, we generally want $k$ and $\Delta$ to be two independent
parameters.  Marking is no longer possible in our general situation.

We briefly describe Moitra's approach before introducing our modifications.  The first observation
is that if the maximum degree is much smaller than the local lemma threshold, variables in the
target distribution are very close to uniform.  As a consequence, if we couple two copies of the
Gibbs distribution while giving different colours at a particular vertex, sequentially and in a
vertex-wise maximal fashion, the discrepancy in the resulting coupling will be logarithmic with high
probability.  Then, one can set up a linear program to do binary search for the marginal
probability, where the variables to solve mimic the transition probabilities in this coupling.  The
marking procedure ensures these locally (almost-)uniform properties to hold at any point of the
coupling process above, by finding a good set of vertices so that we only couple these vertices and
nothing goes awry.

Since marking is no longer possible in our setting, we take an adaptive approach in the coupling
procedure to ensure local (almost-)uniform properties, rather than marking what we are going to
couple in advance.  Although similar in spirit, our proof details are rather different from those by
Moitra \cite{Moi19}.  Since this coupling (or the analysis thereof) is used repeatedly in the whole
algorithm, we have to rework almost all other proofs as well.
A crucial technical contribution of ours is to distinguish two kinds of errors that may 
rise in the linear program.\footnote{These two kinds of errors are not to be confused with the type 
1 and type 2 errors in \cite{Moi19}. Both types are one kind of error in our analysis.}
In particular, the coupling process terminating in logarithmic steps with high probability is \emph{not}
sufficient to bound the number of certain ``bad'' partial colourings and a new exponentially small bound is shown (see Lemma~\ref{lem:outside}).
Moreover, we also streamline the argument and tighten the bounds at various places.  
Hopefully these refinement also sheds some light on where the limit of the method is.

The outline above only gives an approximation of the marginal probabilities.  Due to the lack of
marking, we also need to provide new algorithms for approximate counting and sampling.  For
approximate counting, we use the local lemma again to find a good ordering of the vertices so that
the standard self-reduction goes through.  For sampling, we use the marginal algorithm as an oracle,
to faithfully simulate the true distribution, in an adaptive fashion similar to the coupling
procedure.  At the end of this process, not all vertices will be coloured.  However we show that
with high probability, all remaining connected components have logarithmic sizes and we fill those
in by brutal force enumeration.  The threshold we obtain for sampling is larger than the one for
approximate counting.

\begin{theorem} \label{thm:main2} For integers $\Delta\ge 2$, $k\ge \kone$, and $q>\Sthreshold$,
  there is a sampler whose distribution is $\eps$-close in total variation distance to the uniform
  distribution on all proper colourings, with running time polynomial in the number of vertices and
  $1/\eps$.
\end{theorem}

The correlation decay approach of approximate counting \cite{Wei06,BG08} have been successfully
applied to graph colouring problems \cite{LY13,LYZZ17} or hypergraph problems \cite{BGGGS19}, but it
seems difficult to combine the two in our setting.  More recently, there are other progresses with
respect to approximate counting in the local lemma regime \cite{HSZ19,GJL19,GJ19}.  However, these
results do not directly apply to our situation either.  Indeed, our result can be seen as one step
further to linking the local lemma with approximate counting, as we made Moitra's approach
applicable in a more general setting, where the constraint size does not have to be directly related
to the probability of bad events or the dependency degree.  However, there still seem to be a few
difficulties, such as constraints that cannot be satisfied by partial assignments, to go further
towards the most general abstract setting of the local lemma, and this is an interesting direction
for the future.

The paper is organized as follows.  Section~\ref{sec:prelim} introduces basic notions as well as the
local lemma, and Section~\ref{sec:coupling} introduces the coupling procedure.  We give the
algorithm of estimating marginal probabilities in Section~\ref{sec:marginal}, and use this algorithm
to do counting and sampling in Sections~\ref{sec:count} and~\ref{sec:sampling}, respectively.  To
maintain flexibility, in Sections~\ref{sec:coupling}, \ref{sec:marginal}, \ref{sec:count}, and
\ref{sec:sampling}, we keep track of various parameters, and all parameters are optimized in
Section~\ref{sec:parameters}.  We conclude in Section \ref{sec:conclusion} by describing the
bottleneck of the current approach, and outlining the difficulties for further generalizations.

\section{Preliminary}\label{sec:prelim}

A hypergraph is a pair $H=(V,\+E)$ where $V$ is the collection of vertices and $\+E\subseteq 2^V$ is
the set of hyperedges. We say a hypergraph $H$ is $k$-uniform if every $e\in \+E$ satisfies
$\abs{e}=k$.  Let $q\in\mathbb{N}$ be the number of available colours.  A proper colouring of $H$ is
an assignment $\sigma\in[q]^V$ so that every hyperedge in $\+E$ is not monochromatic, namely that
$\sigma$ satisfies $\abs{\set{\sigma(v)\cmid v\in e}}>1$ for every $e\in\+E$.

Although our goal is to count colourings in $k$-uniform hypergraphs, as the algorithm progresses,
vertices will be pinned to some fixed value.  Therefore we will work with a slightly more general
problem, namely hypergraph colouring with pinnings. Formally, an instance of hypergraph colouring
\emph{with pinnings} is a pair $(H(V,\+E),\+P)$ where $\+P=\set{P_e\subseteq [q] \cmid e\in \+E}$
and $P_e$ is the set of colours that are already present (pinned) inside the edge $e$.  In the
intermediate steps of our algorithms, $\+P$ will be induced by pinning a subset of vertices, but it
is more convenient to consider this slightly more general setup.  For an instance with pinning, a
colouring $\sigma\in [q]^V$ is \emph{proper} if for every $e\in \+E$, it holds that
$\abs{\set{\sigma(v)\cmid v\in e}\cup P_e}>1$.

Denote by $\+C$ the set of all proper colourings of $(H,\+P)$.  For any $\+C'\subseteq \+C$, we use
$\mu_{\+C'}$ to denote the uniform distribution over $\+C'$.  Since there is no weight involved,
$\mu_{\+C}$ is our targeting Gibbs distribution.

Let $\mu$ be a distribution over colourings $\tp{[q]\cup\set{-}}^V$, where ``$-$'' denotes that the
vertex is not coloured (yet).  We say $\mu\tp{\cdot}$ is \emph{pre-Gibbs} with respect to
$\mu_{\+C}$ if for every $\sigma\in\+C$,
\[
  \frac{1}{\abs{\+C}}=\mu_{\+C}(\sigma)=\sum_{\substack{\sigma'\in\tp{[q]\cup\set{-}}^V\\\sigma\models\sigma'}}\mu\tp{\sigma'}\cdot\mu_{\+C}\tp{\sigma\mid
    \sigma'},
\]
where $\sigma\models\sigma'$ means that the full colouring $\sigma$ is consistent with the partial
one $\sigma'$.  In other words, if we draw a partial colouring $\sigma'$ from a pre-Gibbs
distribution $\mu$, and then complete $\sigma'$ uniformly conditioned on coloured vertices (with
respect to $\mu_{\+C}$), the resulting distribution is exactly $\mu_{\+C}$.  Note that in our
definition we do not require the support of $\mu$ to be all partial colourings.

\subsection{Lov\'asz Local Lemma}

Let $(H(V,\+E),\+P)$ be an instance of hypergraph colourings and $q\in\mathbb{N}$ be a non-negative
integer.  We use $\Delta$ to denote the maximum degree of $H$.  Although we consider $k$-uniform
hypergraphs in Theorem \ref{thm:main}, in both the sampling and the counting procedure we will pin
vertices gradually.  Those pinning operations reduce the size of edges,
but in our algorithms we make sure that the size of edges will not go down too much. 
Throughout the section, for every $e\in\+E$, we assume $k'\le \abs{e}\le k$.
Instances of this kind will emerge in Theorem \ref{thm:FPTAS} and Theorem \ref{thm:sampling}.

Let $\-{Lin}(H)$ be the line graph of $H$, that is, vertices in $\-{Lin}(H)$ are hyperedges in $H$
and two hyperedges are adjacent if they share some vertex in $H$.  The ``dependency graph'' of our
problem is simply the line graph of $H$.  For $e\in\+E$, let $\Gamma(e)$ be the neighbourhood of
$e$, namely the set $\left\{e'\mid e\cap e'\neq\emptyset\right\}$.  It is clear that the maximum
degree of $\-{Lin}(H)$ is at most $k(\Delta-1)$.  Hence $\abs{\Gamma(e)}\le k(\Delta-1)$ for any
$e\in \+E$.  With a little abuse of notation, for $v\in V$, let $\Gamma(v)$ be the set of edges in
$\+E$ incident to $v$, i.e., $\Gamma(v)\defeq\set{e\in\+E\cmid v\in e}$.  Furthermore, for any event
$B$ depending a set of vertices $\-{ver}(B)$, let $\Gamma(B)$ be the set of dependent sets of $B$,
i.e., $\Gamma(B)=\left\{e\mid e\cap \-{ver}(B)\neq\emptyset\right\}$.

The (asymmetric) Lov\'asz Local Lemma~(proved by Lov\'asz and published by Spencer \cite{Spe77})
states a sufficient condition for the existence of a proper colouring.  Note that in the following
$\Pr{\cdot}$ refers to the product distribution where every vertex is coloured uniformly and
independently.

\begin{theorem}\label{thm:LLL}
  If there exists an assignment $x:\+E\to(0,1)$ such that for every $e\in\+E$ we have
  \begin{align}\label{eqn:LLL}
    \Pr{\text{$e$ is monochromatic}} \le x(e)\prod_{e'\in\Gamma(e)}\tp{1-x(e')},
  \end{align}
  then a proper colouring exists.
\end{theorem}

When the condition of Theorem~\ref{thm:LLL} is met, we actually have good control over any event in
the uniform distribution $\mu_{\+C}$ due to the next theorem, shown in \cite{HSS11}.

\begin{theorem}\label{thm:LLL-prob}
  If \eqref{eqn:LLL} holds for every $e\in\+E$, then for any event $B$, it holds that
  \[
    \mu_{\+C}(B)\le \Pr{B}\prod_{e\in\Gamma(B)}\tp{1-x(e)}^{-1}.
  \]
\end{theorem}

Theorem~\ref{thm:LLL-prob} also allows us to have some quantitative control over the marginal
probabilities.

\begin{lemma} \label{lem:upbnd} If
  $k'\le \abs{e}\le k$ for any $e\in\+E$, $t\ge k$ and $q\ge \tp{et\Delta}^{\frac{1}{k'-1}}$, then
  for any $v\in V$ and any colour $c\in[q]$,
  \[
    \mPr[\sigma\sim\mu_{\+C}]{\sigma(v)=c} \le \frac{1}{q}\tp{1+\frac{4}{t}}.
  \]
\end{lemma}

\begin{proof}
  Let $x(e)=\frac{1}{t\Delta}$ for every $e\in\+E$.  We first verify that \eqref{eqn:LLL} holds.
  Since $\abs{\Gamma(e)}\le k(\Delta-1)$ and $t\ge k$,
  \begin{align*}
    x(e)\prod_{e'\in\Gamma(e)}\tp{1-x(e')} & \ge \frac{1}{t\Delta} \left( 1- \frac{1}{t\Delta}\right)^{k(\Delta-1)}
                                           & \ge \frac{1}{e t\Delta} \ge q^{1-k'} \ge \Pr{\text{$e$ is monochromatic}}.
  \end{align*}
  Hence, Theorem~\ref{thm:LLL-prob} applies.
  Then,
  \begin{align*}
    \mPr[\sigma\sim\mu_{\+C}]{\sigma(v)=c}
    &\le \frac{1}{q}\tp{1-\frac{1}{t\Delta}}^{-\Delta} \le \frac{1}{q}\exp{\tuple{\frac{2}{t}}} \le \frac{1}{q}\tp{1+\frac{4}{t}}.\qedhere
  \end{align*}
\end{proof}

Unfortunately, Theorem~\ref{thm:LLL-prob} does not give lower bounds directly.  We will instead
bound the probability of blocking $v$ to have colour $c$.

\begin{lemma} \label{lem:lwbnd} If
  $k'\le \abs{e}\le k$ for any $e\in\+E$, $t\ge k$, and $q\ge \tuple{et\Delta}^{\frac{1}{k'-1}}$,
  then for any $v\in V$ and any colour $c\in[q]$,
  \[
    \mPr[\sigma\sim\mu_{\+C}]{\sigma(v)=c} \ge \frac{1}{q}\tuple{1-\frac{1}{t}}.
  \]
\end{lemma}

\begin{proof}
  Fix $v$ and $c$.  For every $e\in\Gamma(v)$, let $\-{Block}_e$ be the event that vertices in $e$
  other than $v$ all have the colour $c$.  Clearly, conditioned on none of $\-{Block}_e$ occurring,
  the probability of $v$ coloured $c$ is larger than $1/q$.  Hence we have that
  \begin{align}\label{eqn:noblock}
    \mPr[\sigma\sim\mu_{\+C}]{\sigma(v)=c} \ge \frac{1}{q}\left( 1 - \sum_{e\in\Gamma(v)} \mu_{\+C}(\-{Block}_e) \right).
  \end{align}  
  
  Clearly $\Pr{\-{Block}_e} = q^{1-\abs{e}} \le q^{1-k'}$.  Again let $x(e) = \frac{1}{t\Delta}$ for
  every $e\in\+E$ and \eqref{eqn:LLL} holds.  Since $\abs{\Gamma(\-{Block}_e)}\le k(\Delta-1)+1$ and
  $t\ge k$, by Theorem \ref{thm:LLL-prob},
  \begin{align}
    \mu_{\+C}(\-{Block}_e) & \le q^{1-k'} \left( 1-\frac{1}{t\Delta} \right)^{-k(\Delta-1)-1} \le \frac{1}{t\Delta}. \label{eqn:LLL-block}
  \end{align}

  Plugging \eqref{eqn:LLL-block} into \eqref{eqn:noblock} yields
  \begin{align*}
    \mPr[\sigma\sim\mu_{\+C}]{\sigma(v)=c} & \ge \frac{1}{q}\left( 1-\frac{1}{t} \right).\qedhere
  \end{align*}
\end{proof}

Combining Lemma~\ref{lem:upbnd} and Lemma~\ref{lem:lwbnd}, we obtain the following result.

\begin{lemma}\label{lem:LLL-prob}
  If
  $k'\le \abs{e}\le k$ for any $e\in\+E$, $t \ge k$ and $q\ge \tp{et\Delta}^{\frac{1}{k'-1}}$, then
  for any $v\in V$ and any colour $c\in[q]$,
  \[
    \frac{1}{q}\tp{1-\frac{1}{t}} \le \Pr[\sigma\sim \mu_{\+C}]{\sigma(v)=c} \le
    \frac{1}{q}\tp{1+\frac{4}{t}}.
  \]
\end{lemma}

\section{The coupling}\label{sec:coupling}

Recall that a partial colouring is an assignment $\sigma\in\tp{[q]\cup\set{-}}^V$ where ``$-$''
denotes an unassigned colour. Fix a vertex $v\in V$ and two distinct colours $c_1,c_2\in[q]$, we
define two initial partial colourings $X_0$ and $Y_0$ that assign $v$ with colours $c_1$ and $c_2$
respectively and let all other vertices be unassigned. We use $\+C_1$ and $\+C_2$ to denote the set
of proper colourings with $v$ fixed to be $c_1$ and $c_2$ respectively. For a partial colouring $X$,
we use $\+C_X$ to denote the set of proper colourings consistent with $X$.

Moitra~\cite{Moi19} introduced the following intriguing idea (in the setting of CNF) to compute the
ratio of marginal probabilities on $v$.  Couple $\mu_{\+C_1}$ and $\mu_{\+C_2}$ in a sequential way.
Start from $v$, where the colours differ, and proceed in a breadth-first search manner, vertex by
vertex.  At each vertex we draw a colour from $\mu_{\+C_1}$ and $\mu_{\+C_2}$, respectively,
conditioned on all the existing colours, and couple them maximally.  The process ends when the set
of vertices coupled successfully form a cut separating $v$ from uncoloured vertices.  If every
vertex we encounter has its marginal distribution close enough to the uniform distribution, then
this coupling process terminates quickly with high probability.  These local almost-uniform
properties are guaranteed by Lemma \ref{lem:LLL-prob}.  Then Moitra sets up a clever linear program
(LP), where the variables mimic transition probabilities during the coupling (but in some
conditional way), and shows that the LP is sufficient to recover the marginal distribution at $v$ by
a binary search.

We apply the same idea here for hypergraph colourings.  However, one needs to carefully implement
the coupling to guarantee that all marginal distributions encountered are close enough to uniform.
Formally, we describe our coupling process in Algorithm~\ref{alg:coupling}.  The coupling process
applies to hypergraphs with edge size between $k_1$ and $k$ for some parameter $0<k_1\le k$. There
is another parameter $0<k_2<k_1$ and all these parameters will be set in
Section~\ref{sec:parameters}.  The output is a pair of partial colourings $(X,Y)$ extending $X_0$
and $Y_0$ respectively.  Notice that in order to implement the coupling process, we fix an arbitrary
ordering of edges and vertices in advance.

\begin{algorithm}[htbp]
  \caption{The coupling process}
  \label{alg:coupling}
  \begin{algorithmic}[1]
    \State \textbf{Input:} \emph{A hypergraph $H(V,\+E)$ with pinnings $\+P$ and
      $k_1\le\abs{e}\le k$ for every $e\in\+E$, two partial colourings $X_0$ and $Y_0$.}%
    \State \textbf{Output:} \emph{$V_{\mathrm{col}}\subseteq V$, a partition $V_1\sqcup V_2=V$, and
      two partial colourings $X, Y$ defined on $V_{\mathrm{col}}$.}%
    \State $V_1\gets\set{v}$, $V_2\gets V\setminus V_1$, $V_{\mathrm{col}}\gets \set{v}$;%
    \State $X\gets X_0$, $Y\gets Y_0$;%
    \While {$\exists e\in\+E$ s.t. $e\cap V_1\ne\varnothing$ and $e\cap V_2\ne\varnothing$}%
    \State Let $e$ be the first such hyperedge;%
    \State Let $u$ be the first vertex in $e\cap V_2$;%
    \State Sample a pair of colours $(c_x,c_y)$ according to the maximal coupling of the marginal
    distribution at $u$ conditioned on $X$ and $Y$ respectively;%
    \State Extend $X$ and $Y$ by colouring $u$ with $c_x$ and $c_y$, respectively;%
    \State $V_{\mathrm{col}}\gets V_{\mathrm{col}}\cup\set{u}$;%
    \If {$c_x\ne c_y$} \State $V_1\gets V_1\cup\set{u}$, $V_2\gets V_2\setminus\set{u}$;%
    \EndIf \For {$e\in\Gamma(u)\cap\+E$ s.t.~$e$ is satisfied by both $X$ and $Y$} \State
    $\+E\gets \+E\setminus \set{e}$;%
    \EndFor \For {$e\in\Gamma(u)\cap\+E$ s.t.~$e\cap V_1\ne\varnothing$, $e\cap V_2\ne\varnothing$,
      and $\abs{e\cap V_{\mathrm{col}}}= k_2$} \State
    $V_1\gets V_1\cup (e \setminus V_{\mathrm{col}})$, $V_2\gets V\setminus V_1$;%
    \State $\+E\gets \+E\setminus \set{e}$; \EndFor \EndWhile
  \end{algorithmic}
\end{algorithm}

The set $V_{\-{col}}$ consists of all coloured vertices.  Intuitively, the set $V_1$ contains
vertices that have failed the coupling and $V_2$ is its complement.  Once a hyperedge is satisfied
by both partial colourings $X$ and $Y$, it has no effect any more and is thus removed.

The main difference from Moitra's coupling \cite{Moi19} is that we cannot choose what vertices to
couple in advance (``marking'').  Instead, we take an adaptive approach to ensure that no hyperedge
becomes too small.  Once $k_2$ vertices of a hyperedge are
coloured, 
all the rest vertices are considered ``failed'' in the coupling (namely they are added to $V_1$).
However these failed vertices are left uncoloured.

Algorithm \ref{alg:coupling} outputs a pair of partial colourings $X,Y$ defined on $V_{\-{col}}$ and
a partition of vertices $V=V_1\sqcup V_2$.  For any edge $e$ in the original $\+E$ such that
$e\cap V_1\neq\varnothing$ and $e\cap V_2\neq\varnothing$, it is removed because either it is
satisfied by both $X$ and $Y$, or $k_2$ vertices in $e$ have been coloured.  In the latter case, all
vertices in $e$ are either coloured or in $V_1$, namely $e\subset V_1\cup V_{\-{col}}$.  Hence all
edges intersecting $V_1$ and $V_2\setminus V_{\-{col}}$ are satisfied by both $X$ and $Y$.  This
fact will be useful later.

For $u\in V$, let $\Gamma_{\-{ver}}(u)$ denote the neighbouring vertices of $u$ (including $u$),
namely $\Gamma_{\-{ver}}(u) = \left\{w\mid \exists e\in\+E, \{u,w\}\subseteq e\right\}$, and let
$\Gamma_{\-{ver}}(U) = \bigcup_{u\in U}\Gamma_{\-{ver}}(u)$ for a subset $U\subseteq V$.  The
following lemma summarizes some properties of this random process.

\begin{lemma}\label{lem:coupling}
  The following properties of Algorithm \ref{alg:coupling} hold:
  \begin{enumerate}
  \item All coloured vertices are either in $V_1$ or incident to $V_1$, namely
    $V_{\-{col}}\subseteq \Gamma_{\-{ver}}(V_1)$;
  \item The distributions of $X$ and $Y$
    are pre-Gibbs with respect to $\mu_{\+C_1}$ and $\mu_{\+C_2}$ respectively.
  \end{enumerate}
\end{lemma}
\begin{proof}
  For (1), notice that whenever we add a vertex $u$ into $V_{\-{col}}$, it must hold that $u\in e$
  for some $e \cap V_1\ne \varnothing$ at the time.  The claim follows from a simple induction.

  For (2), we only prove the lemma for $X$.  The proof for $Y$ is similar.  The partial colouring
  $X$ is generated in the following way: at each step either the process ends, or the next
  uncoloured vertex $u$ is chosen and extend $X$ to $u$ with the correct (conditional) marginal
  probability and repeat.  Our decisions (whether or not to halt, and what is the next $u$) depend
  on $Y$ in addition to the partial colouring $X$ so far.

  An intermediate state $\+S$ of Algorithm \ref{alg:coupling} consists of partial colourings $X$,
  $Y$, $V_{\-{col}}$, and $V_1$.\footnote{We note that actually $V_{\-{col}}$ and $V_1$ are
    completely determined by $X$ and $Y$, but we do not need this fact here. The reason for
    $V_{\-{col}}$ is obvious, and $V_1$ can be deduced from $X,Y$ by simulating the whole process
    from start.}  Our claim is that, conditioned on any valid $\+S$, the distribution of the final
  output (on the $X$ side) of Algorithm \ref{alg:coupling} is pre-Gibbs with respect to
  $\mu_{\+C_{X}}$.  The lemma clearly follows from the claim by setting $\+S$ to the initial state
  of Algorithm \ref{alg:coupling}.

  We induct on the maximum possible future steps of $\+S$.  The base case is that $\+S$ will halt
  immediately.  Thus the output is simply $X$ and completing it yields the uniform distribution on
  $\+C_{X}$.  That is, the output is pre-Gibbs.

  For the induction step, $\+S$ will not halt but rather, extend the colourings to some vertex $u$
  which is deterministically selected by our algorithm. Let $\tau_{\+S}(\cdot)$ denote the
  measure on colourings obtained by completing the output of Algorithm \ref{alg:coupling}
  conditioned on $\+S$.  Let $X^{u\gets c}$ be a partial colouring defined on $V_{\-{col}}\cup\{u\}$
  by extending $X$ to $u$ with colour $c$, and $\+S'$ be an internal state consistent with
  $X^{u\gets c}$, denoted by $\+S'\models X^{u\gets c}$ .  Moreover, let $q(\+S')$ be the
  probability of transiting from $\+S$ to $\+S'$.  Since the marginal probability at $u$ only
  depends on the previous partial colourings $X'$, we have that
  \begin{align}\label{eqn:independence-of-u}
    \sum_{\+S'\models X^{u\gets c}}q(\+S') = \mu_{C_{X}}(X^{u\gets c}),
  \end{align}
  where $\mu_{C_{X}}(X^{u\gets c})$ is in fact the marginal probability of the colour $c$ at $u$
  conditioned on $X$.  By our induction hypothesis, conditioned on $\+S'$, the final output is
  pre-Gibbs with respect to $\+C_{X^{u\gets c}}$.  That is,
  \begin{align}\label{eqn:induction-hypothesis}
    \tau_{\+S'}(\cdot) = \mu_{\+C_{X^{u\gets c}}}(\cdot).
  \end{align}
  For $\sigma\in\+C_{X}$, suppose $X^{u\gets c}$ is the partial colouring of $\sigma$ restricted to
  $V_{\-{col}}\cup\{u\}$.  Then we have that
  \begin{align*}
    \tau_{\+S}(\sigma) & = \sum_{\+S'\models X^{u\gets c}}q(\+S')\tau_{\+S'}(\sigma)\\
                       & = \sum_{\+S'\models X^{u\gets c}}q(\+S') \mu_{\+C_{X^{u\gets c}}}(\sigma) \\
                       & = \mu_{\+C_{X^{u\gets c}}}(\sigma) \sum_{\+S'\models X^{u\gets c}}q(\+S') \\
                       & = \mu_{\+C_{X^{u\gets c}}}(\sigma) \mu_{C_{X}}(X^{u\gets c}) \\
                       & = \mu_{\+C_{X}}(\sigma),
  \end{align*}
  where in the second line we use \eqref{eqn:induction-hypothesis}, and in the fourth line we use
  \eqref{eqn:independence-of-u}.  The claim follows.
\end{proof}

Therefore, the output of Algorithm~\ref{alg:coupling} is a coupling of two pre-Gibbs measures such
that they are defined on the same set of vertices $V_{\-{col}}$. We use $\mu_{\-{cp}}(\cdot,\cdot)$
to denote this joint distribution.

It is possible to show that the final size of $\abs{V_1}$ is $O(\log \abs{V})$ with high
probability.  This fact will not be directly used, and is indeed not strong enough for the algorithm
and its analysis in the next section.  We will omit its proof.  What we will show eventually is
that, conditioned on a randomly chosen colouring from $\+C_1$ or $\+C_2$, the probability that the
coupling process terminates decays exponentially with the depth.  There are two levels of randomness
here, and they will be separated, since the linear program later will only be able to certify the
second kind randomness.

Later, in Section \ref{sec:sampling}, when we do sampling, we will consider a similar procedure,
Algorithm~\ref{algo:sampling}, and we will show that the connected components produced by
Algorithm~\ref{algo:sampling} are $O(\log \abs{V})$ with high probability (Lemma \ref{lem:fail}).
This is in the same vein as $\abs{V_1}$ being size $O(\log \abs{V})$ with high probability in
Algorithm~\ref{alg:coupling}.

\section{Computing the marginals}\label{sec:marginal}

In the previous section, we introduced a random process to generate a joint distribution of partial
colourings $\mu_{\-{cp}}(\cdot,\cdot)$, whose marginal distributions are pre-Gibbs.  Recall that we
fixed $X(v)=c_1$ and $Y(v)=c_2$. Let $\-q_i$ denote the marginal probability in $\mu_{\+C}$ of $v$
being coloured by $c_i$, for $i=1,2$.  That is, $\-q_i=\frac{\abs{\+C_i}}{\abs{\+C}}$ for $i=1,2$.
The coupling naturally induces an (imaginary) sampler to uniformly sample from $\+C_1\cup\+C_2$ as
follows:

\begin{itemize}
\item [\bf Step 1:] Sample $(X,Y)=(x,y)$ using Algorithm \ref{alg:coupling};
\item [\bf Step 2:] Let $v\gets c_1$ with probability $\frac{\-q_1}{\-q_1+\-q_2}$ and $v\gets c_2$
  otherwise;
\item [\bf Step 3:] If $v$ is coloured by $c_1$, uniformly output a colouring in $\+C_{x}$,
  otherwise uniformly output a colouring in $\+C_{y}$.
\end{itemize}

We denote this sampler by $\mathbb{S}$.  The output of $\mathbb{S}$ is uniform over $\+C_1\cup\+C_2$
is because by Lemma \ref{lem:coupling}, the output distribution of Algorithm \ref{alg:coupling},
projected to either side, is pre-Gibbs.  Then we choose the final colouring proportional to the
correct ratio.

One can represent the coupling process (Algorithm \ref{alg:coupling}) as traversing a
(deterministic) \emph{coupling tree} $\+T$ constructed as follows: each vertex in $\+T$ represents a
pair of partial colourings $(x,y)$\footnote{We use small letters $x,y$ to denote particular partial
  colourings, and reserve capital $X,Y$ to denote random ones.} defined on some $V_{\-{col}}$ that
have appeared in the coupling.  We write $(x,y)\in\+T$ if $(x,y)$ is a pair of partial colourings
represented by some vertex in $\+T$.  Although the intermediate state of Algorithm
\ref{alg:coupling} consists of partial colourings $x,y$ together with $V_{\-{col}}$ and $V_1$, we
can actually deduce $V_{\-{col}}$ from $x,y$, as well as $V_1$ by simulating Algorithm
\ref{alg:coupling} from the start given $x$ and $y$.  Thus the pair $(x,y)$ determines either that
the coupling should halt or the next vertex $u$ to extend to.  In the coupling tree $\+T$, $(x,y)$
either is a leaf or has $q^2$ children, which correspond to the $q^2$ possible ways to extend
$(x,y)$ by colouring $u$.  The root of the tree is the initial pair $(x_0,y_0)$ defined on
$\set{v}$.

In the following, we identify a collection of conditional marginal probabilities that keeps the
information of the coupling process.

First, consider a pair of partial colourings $(x,y)\in\+T$ which is a leaf, and any two proper
colourings $\sigma_x,\sigma_y$ such that $\sigma_x\models x$ and $\sigma_y\models y$.  In the
probability space induced by the sampler introduced above, define
\begin{align*}
  \-p_{x,y}^x&\defeq \Pr[(X,Y)\sim\mu_{\-{cp}}]{X=x,Y=y\mid \mathbb{S}\mbox{ outputs }\sigma_x};\\
  \-p_{x,y}^y&\defeq \Pr[(X,Y)\sim\mu_{\-{cp}}]{X=x,Y=y\mid \mathbb{S}\mbox{ outputs }\sigma_y}.
\end{align*}
These quantities are well defined and independent of the particular choices of $\sigma_x$ and
$\sigma_y$.  Essentially we only condition on the random choice at step 2 of $\mathbb{S}$.  Once
that choice is made, the output is uniform over $\+C_x$ or $\+C_y$.

Perhaps a clearer way of seeing this independence is to give more explicit expressions to
$p_{x,y}^x$ and~$p_{x,y}^y$.  By Bayes' rule,
\begin{align}
  \-p_{x,y}^x
  &=
    \frac{
    \Pr[(X,Y)\sim\mu_{\-{cp}}]{\mathbb{S}\mbox{ outputs }\sigma_x\mid X=x,Y=y}
    \mu_{\-{cp}}(x,y)}{\Pr{\mathbb{S}\mbox{ outputs }\sigma_x}} 
    \nonumber \\
  &=\-q_1\cdot\frac{\abs{\+C_1\cup\+C_2}}{\abs{\+C_x}}\cdot\mu_{\-{cp}}(x,y);\label{eqn:BayesX}\\
  \-p_{x,y}^y
  &=
    \frac{
    \Pr[(X,Y)\sim\mu_{\-{cp}}]{\mathbb{S}\mbox{ outputs }\sigma_y\mid X=x,Y=y}
    \mu_{\-{cp}}(x,y)}{\Pr{\mathbb{S}\mbox{ outputs }\sigma_y}} 
    \nonumber\\
  &=\-q_2\cdot\frac{\abs{\+C_1\cup\+C_2}}{\abs{\+C_y}}\cdot\mu_{\-{cp}}(x,y).\label{eqn:BayesY}
\end{align}
Combining two identities above we obtain
\begin{align}\label{eq:identity}
  \-q_1\cdot \-p_{x,y}^y\cdot\abs{\+C_y}&=\-q_2\cdot\-p_{x,y}^x\cdot\abs{\+C_x}.  
\end{align}

A crucial observation is that, for every pair of partial colourings $(x,y)$ that is a leaf of $\+T$
with corresponding $V_{\-{col}},V_1,V_2$, the ratio $\frac{\abs{\+C_{x}}}{\abs{\+C_{y}}}$ can be
computed in $q^{\abs{V_1\setminus V_{\-{col}}}}$ time.  This is because when
Algorithm~\ref{alg:coupling} terminates, all edges intersecting $V_1$ and $V_2\setminus V_{\-{col}}$
are satisfied by both $x$ and $y$.  The numbers of ways colouring blank vertices in $V_2$ cancel
out, and we only need to enumerate all colourings for blank vertices inside $V_1$.  Let
$r_{x,y}=\frac{\abs{\+C_{x}}}{\abs{\+C_{y}}}$.

Next, consider an internal $(x,y)$ in the coupling tree $\+T$. We interpret $\-p_{x,y}^x$ and
$\-p_{x,y}^y$ as the probability that the coupling process has ever arrived at an internal pair of
partial colourings $(x,y)$ conditioned on the output of $\mathbb{S}$ being $\sigma_x$ and $\sigma_y$
for any $\sigma_x,\sigma_y$ such that $\sigma_x\models x$ and $\sigma_y\models y$, respectively.
Note that the definition is consistent with our previous definition when $(x,y)$ is a leaf of $\+T$.
Recall that $(x_0,y_0)$ is the root of $\+T$, namely $x_0$ or $y_0$ only colours $v$ with $c_1$ or
$c_2$, respectively.  For $(x_0,y_0)$, we have that
\begin{align}\label{eq:root}
  p_{x_0,y_0}^{x_0}=p_{x_0,y_0}^{y_0}=1.
\end{align}
Moreover, for an internal $(x,y)$ whose children are defined on
$V_{\-{col}}'=V_{\-{col}}\cup\set{u}$, it holds that
\begin{align}
  \text{for every $c\in[q]$, }\-p_{x,y}^x&=\sum_{c'\in[q]}\-p_{x^{u\gets c},y^{u\gets c'}}^{x^{u\gets c}};\label{eq:int-conX}\\
  \text{for every $c\in[q]$, }\-p_{x,y}^y&=\sum_{c'\in[q]}\-p_{x^{u\gets c'},y^{u\gets c}}^{y^{u\gets c}}.\label{eq:int-conY}
\end{align}
where we use $x^{u\gets c}$ to denote the partial colouring that extends $x$ by assigning colour $c$
to the vertex~$u$. To see why~\eqref{eq:int-conX} holds, we note that conditioned on the event that
$\mathbb{S}$ outputs some $\sigma_x\models x$, the colouring on $u$ is $\sigma_x(u)$ and all the
randomnesses are from the choice of colours in $Y$ on $u$. The identity~\eqref{eq:int-conY} holds
for the same reason, by reversing the roles of $x$ and $y$.

In fact, when the coupling process is at some internal node of the coupling tree, say $(x,y)$,
defined on $V_{\-{col}}$, and the next step is to sample the colour on a vertex $u$, one can recover
the distribution of the colour on $u$ in the next step from the values
\[
  \set{\-p^{x^{u\gets c}}_{x^{u\gets c},y^{u\gets c'}},\-p^{y^{u\gets c}}_{x^{u\gets c'},y^{u\gets
        c}}\cmid c,c'\in[q]}
\]
by solving linear constraints using Bayes' rule.  Therefore, the collection
$\set{\-p_{x,y}^x,\-p_{x,y}^y\cmid (x,y)\in\+T}$ encodes all information of the coupling process.

\subsection{The linear program}

\newcommand{\ratiolb}{\underline{r}} \newcommand{\ratioub}{\overline{r}}

The values $\-p_{x,y}^x$ and $\-p_{x,y}^y$ are unknown and we are going to impose a few necessary
linear constraints on them.  The basic constraints are derived from \eqref{eq:identity},
\eqref{eq:root}, \eqref{eq:int-conX}, and~\eqref{eq:int-conY}.  To this end, for every node $(x,y)$
in $\+T$, we introduce two variables $\ul{p}_{x,y}^x$ and $\ul{p}_{x,y}^y$, aiming to mimic
$\-p_{x,y}^x$ and $\-p_{x,y}^y$.

The full coupling tree $\+T$ is too big, and we will truncate it up to some depth $L>0$.  The
quantity $L$ will be set later.  We will perform a binary search to estimate the ratio
$\frac{\-q_1}{\-q_2}$ using the truncated coupling tree.  Thus, we introduce two variables
$\ratioub$ and $\ratiolb$ as our guesses for upper and lower bounds of $\frac{\-q_1}{\-q_2}$.  Let
$\+T_L$ be the coupling tree truncated at depth $L$, and denote by $\+L(\+T)$ the leaves of a tree
$\+T$.  Since the coupling procedure colours one vertex at a time, for any node $(x,y)\in \+T_L$, we
have that $\abs{V_{\-{col}}}\le L$ where $V_{\-{col}}$ is determined by $(x,y)$.  Formally, we have
three types of constraints.

\bigskip
\noindent\textbf{Constraints 1}: For every leaf $(x,y)\in \+L(\+T_L)$ with corresponding
$\abs{V_{\-{col}}}<L$, we have the constraints:
\begin{align*}
  \ratiolb\cdot \ul{p}_{x,y}^y & \le \ul{p}_{x,y}^x\cdot r_{x,y};\\
  \ul{p}_{x,y}^x\cdot r_{x,y} &\le \ratioub\cdot \ul{p}_{x,y}^y;\\
  0\le \ul{p}_{x,y}^x,\,&\ul{p}_{x,y}^y\le 1.
\end{align*}
Constraints 1 are relaxed versions of identity \eqref{eq:identity}.  It will be clear soon that
these constraints are the most critical ones, as they guarantee that we can recover the marginal
probability on $v$ from these variables.  However, in order to compute $r_{x,y}$, one needs
$\exp(L)$ amount of time.  This forces us to truncate at only logarithmic depth in the coupling tree in order
to get a polynomial time algorithm, but we will show later that this is enough.

\bigskip
\noindent\textbf{Constraints 2}:
For the root $(x_0,y_0)\in \+T$, we have
\[
  \ul{p}_{x_0,y_0}^{x_0}=\ul{p}_{x_0,y_0}^{y_0}=1.
\]
Moreover, for every non-leaf $(x,y)\in\+T$ with corresponding $\abs{V_{\-{col}}}< L$, let $u$ be the
next vertex to couple.  We have the following constraints:
\begin{align*}
  \mbox{for every $c\in[q]$, }\ul{p}_{x,y}^x&=\sum_{c'\in[q]}\ul{p}_{x^{u\gets c},y^{u\gets c'}}^{x^{u\gets c}};\\
  \mbox{for every $c\in[q]$, }\ul{p}_{x,y}^y&=\sum_{c'\in[q]}\ul{p}_{x^{u\gets c'},y^{u\gets c}}^{y^{u\gets c}};\\
  0\le \ul{p}_{x,y}^x,\,&\ul{p}_{x,y}^y\le 1.
\end{align*}
These constraints faithfully realize the properties \eqref{eq:root}, \eqref{eq:int-conX}, and
\eqref{eq:int-conY}.

\newcommand{\tvalue}{\ensuremath{5\tp{e^2k^3\Delta^{3}}^{\frac{1}{1-\beta}}}}

\bigskip
\noindent\textbf{Constraints 3}: For every $c,c'\in[q]$ that $c\ne c'$, we add constraints:
\begin{align*}
  \ul{p}_{x^{u\gets c},y^{u\gets c'}}^{x^{u\gets c}}  &\le \frac{5}{t^*}\cdot \ul{p}_{x,y}^x;\\
  \ul{p}_{x^{u\gets c},y^{u\gets c'}}^{y^{u\gets c'}} &\le \frac{5}{t^*}\cdot \ul{p}_{x,y}^y.
\end{align*}
We will eventually set $t^*=\tvalue$ in Lemma~\ref{lem:inside}, where the parameter $0<\beta<1$ will
become clear in Definition~\ref{def:badcolouring}.

These constraints reflect the fact that the coupling at individual vertices is very likely to
succeed, due to Lemma~\ref{lem:LLL-prob}.  Assume the conditions of Lemma~\ref{lem:LLL-prob} are met
with $t=t^*$.  We claim the following property of those true values $\{\-p_{x,y}^x\}$.
\begin{claim}
  \[
    \frac{\-p_{x^{u\gets c},y^{u\gets c}}^{x^{u\gets c}}}{\-p_{x,y}^x}  \ge 1-\frac{5}{t^*}.
  \]
\end{claim}

The claim implies that these true values satisfy \textbf{Constraints 3} since they also satisfy
\textbf{Constraints 2}.  We use \eqref{eqn:BayesX} to show the claim.  By Lemma~\ref{lem:LLL-prob},
\begin{align*}
  \frac{\abs{C_{x}}}{\abs{C_{x^{u\gets c}}}} = \frac{1}{\mPr[\sigma\sim \mu_{\+C_x}]{\sigma(u)=c}} \ge \frac{qt^*}{t^*+4}.
\end{align*}
Again by Lemma~\ref{lem:LLL-prob}, the coupling at $u$ with any colour $c$ succeeds with probability
at least $\frac{1}{q}\left( 1-\frac{1}{t^*} \right)$.  Thus the ratio
$\frac{\mu_{\-{cp}}(x^{u\gets c},y^{u\gets c})}{\mu_{\-{cp}}(x,y)}$, which can be viewed as the
probability of coupling $u$ successfully with colour $c$ conditioned on reaching $(x,y)$, is at
least $\frac{1}{q}\left( 1-\frac{1}{t^*} \right)$.  Combine these facts with \eqref{eqn:BayesX},
\begin{align*}
  \frac{\-p_{x^{u\gets c},y^{u\gets c}}^{x^{u\gets c}}}{\-p_{x,y}^x} 
  = \frac{\abs{C_{x}}}{\abs{C_{x^{u\gets c}}}} \cdot \frac{\mu_{\-{cp}}(x^{u\gets c},y^{u\gets c})}{\mu_{\-{cp}}(x,y)}  \ge \frac{qt^*}{t^*+4}\cdot \frac{1}{q}\left( 1-\frac{1}{t^*} \right) = 1 - \frac{5}{t^*+4} \ge 1 - \frac{5}{t^*}.
\end{align*}
Similar inequalities hold for $\{p_{x,y}^y\}$ due to \eqref{eqn:BayesY}.

\subsection{Analysis of the LP}

In this subsection, we show that the LP can be used to obtain an efficient and accurate estimator of
marginals.

\newcommand{\badfraction}{e^{-\ell}}
\newcommand{\badthreshold}{\ensuremath{C\Delta^{\frac{3}{\beta(k_{2}-1)}}}}
\newcommand{\couplingfailure}{e^{-\ell}}
\newcommand{\couplingthreshold}{\ensuremath{C\Delta^{\frac{4-\beta}{(1-\beta)(k_1-k_2-1)}}}}
\newcommand{\gammavalue}{\ensuremath{4e^{-\frac{L}{k^3\Delta^2}}}}
\newcommand{\ellvalue}{\ensuremath{\frac{L}{k^3\Delta^2}}}
\newcommand{\marginalthreshold}{\ensuremath{\max\left\{\left( ek\Delta
      \right)^{\frac{1}{k_1-2}},\beta^{\frac{-1}{k_2-1}},\badthreshold,\couplingthreshold\right\}}}

\begin{theorem}\label{thm:marginal}
  Let $\Delta\ge 2$ and $k>0$ be two integers.  Let $0<\beta<1$ be a constant.  Let $0<k_2<k_1\le k$
  be integers.  Let $H=(V,\+E)$ be a hypergraph with pinnings $\+P$, maximum degree $\Delta$ such
  that $k_1\le \abs{e}\le k$ for every $e\in \+E$.  If \[ q >\marginalthreshold \] where
  \begin{align*}
    C>\max\left\{\left(\frac{e^{\beta+3}k^3}{\beta^\beta}\cdot\binom{k}{k_2}\right)^{\frac{1}{\beta(k_2-1)}}, \left(5e\tp{e^2k^3}^{\frac{1}{1-\beta}}\right)^{\frac{1}{k_1-k_2-1}}\right\},
  \end{align*}
  then there is a deterministic algorithm that, for every $v\in V$, $c\in[q]$ and $\eps > 0$, it
  computes a number $\widehat{p}$ satisfying
  \[
    e^{-\eps}\cdot \widehat{p}\le \Pr[\sigma\sim\mu_{\+C}]{\sigma(v)=c}\le e^{\eps}\cdot
    \widehat{p}.
  \]
  in time $\poly(\frac{1}{\eps})$.
\end{theorem}

Before diving into the proof details, let us first imagine that we set up the LP for the whole
coupling tree.  To do this would require exponential amount of time, but we show that this indeed
can be used to estimate the marginals to arbitrary precision.  We use
$\set{\widehat{p}_{x,y}^x,\widehat{p}_{x,y}^y}_{(x,y)\in\+T}$ to denote a solution of this LP. Due to
\textbf{Constraints 2}, a simple induction shows that for every $L\le \abs{V}$ and $\sigma\in\+C_1$,
\begin{align*}
  \sum_{(x,y)\in\+L(T_L):\;\sigma\models x} \widehat{p}_{x,y}^{x} = 1.
\end{align*}
In particular, when $L=\abs{V}$, this means that
\begin{align*}
  \sum_{(x,y)\in\+L(T):\;\sigma\models x} \widehat{p}_{x,y}^{x} = 1.
\end{align*}
Similar equalities hold on the $Y$ side.  Using this, we rewrite the ratio
$\frac{\abs{\+C_1}}{\abs{\+C_2}}$ as follows:
\begin{align*}
  \frac{\abs{\+C_1}}{\abs{\+C_2}}
  & =\frac{\sum_{\sigma\in \+C_1}\,1}{\sum_{\sigma\in\+C_2}\,1}
    =\frac{\sum_{\sigma\in\+C_1}\sum_{(x,y)\in\+L(\+T): \sigma\models x}\widehat p_{x,y}^x}
    {\sum_{\sigma\in\+C_2}\sum_{(x,y)\in\+L(\+T): \sigma\models y}\widehat p_{x,y}^y}\\
  & =\frac{\sum_{(x,y)\in\+L(\+T)}\sum_{\sigma\models x}\widehat p_{x,y}^x}
    {\sum_{(x,y)\in\+L(\+T)}\sum_{\sigma\models y}\widehat p_{x,y}^y}\\
  & =\frac{\sum_{(x,y)\in\+L(\+T)}\widehat p_{x,y}^x\abs{C_x}}
    {\sum_{(x,y)\in\+L(\+T)}\widehat p_{x,y}^y\abs{C_y}}.
\end{align*}
Recall $r_{x,y}=\frac{\abs{C_x}}{\abs{C_y}}$.  By \textbf{Constraints 1}, we know that for any
$(x,y)\in\+L(\+T)$,
\begin{align*}
  \ratiolb \le \frac{\widehat p_{x,y}^x\abs{C_x}}{\widehat p_{x,y}^y\abs{C_y}} \le \ratioub.
\end{align*}
It implies that
\begin{align*}
  \ratiolb \le \frac{\abs{\+C_1}}{\abs{\+C_2}} \le \ratioub.
\end{align*}

Unfortunately, as the size and the computational cost of setting up the LP is exponential in $L$, we
have to truncate the tree at a suitable place.  The rest of our task is to show that the error caused by the truncation
is small.  One may notice that in the analysis above we do not use \textbf{Constraints $3$}.
Indeed, these constraints are used to bound the truncation error.

Intuitively, the truncation error comes from the proper colourings so that the coupling does not
halt at depth $L$ (since we cannot impose \textbf{Constraints 1} for these nodes).  A naive approach
would then try to show that conditioned on any proper colouring as the final output, the coupling
will terminate quickly.  This is unfortunately not true and there exist ``bad'' colourings so that
the coupling does not terminate at level $L$ with high probability.  For example, given a pre-determined ordering
of vertices and edges, a proper colouring $\sigma\in\+C_1$ may render all vertices encountered in
Algorithm \ref{alg:coupling} with the same colour.  Hence conditioned on this $\sigma$ on the $X$
side, Algorithm \ref{alg:coupling} will not stop until all edges are enumerated.

We will show, nonetheless, that the fraction of ``bad'' colourings is small.  Let us formally define
bad colourings first.  We need to use the notion of $\{2,3\}$-trees.  This notion dates back to
Alon's parallel local lemma algorithm \cite{Alo91}.

\begin{definition}[$\{2,3\}$-tree]
  Let $G=(V,E)$ be a graph.  A set of vertices $T\subseteq V$ is a \emph{$\{2,3\}$-tree} if (1) for
  any $u,v\in T$, $\-{dist}_G(u,v)\ge 2$; (2) if one adds an edge between every $u,v\in T$ such that
  $\-{dist}_G(u,v)=2$ or $3$, then $T$ is connected.
\end{definition}

%
We will need to count the number of $\{2,3\}$-trees later for union bounds.  The following lemma,
due to Borgs et al. \cite{BCKL13}, counts the number of connected induced subgraphs in a graph.

\begin{lemma}\label{lem:subnumber}
  Let $G=(V,E)$ be a graph with maximum degree $d$ and $v\in V$ be a vertex. The number of connected
  induced subgraphs of size $\ell$ containing $v$ is at most $\frac{(ed)^{\ell-1}}{2}$.
\end{lemma}

\begin{corollary}\label{cor:2-tree-number}
  Let $G=(V,E)$ be a graph with maximum degree $d$ and $v\in V$ be a vertex. Then the number of
  $\{2,3\}$-trees in $G$ of size $\ell$ containing $v$ is at most $\frac{\tp{ed^3}^{\ell-1}}{2}$.
\end{corollary}
\begin{proof}
  Let $G'=(V,E')$ be the graph with vertex set $V$ and $(u,v)\in E'$ if $\-{dist}_G(u,v)=2,3$. The
  degree of $G'$ is at most $d^3$ and any $\{2,3\}$-tree in $G$ is a connected set of vertices in
  $G'$. Therefore, the number of $\{2,3\}$-trees in $G$ containing $v$ of size $\ell$ can be bounded
  by the number of induced subgraphs in $G'$ containing $v$ of size
  $\ell$. Lemma~\ref{lem:subnumber} then concludes the proof.
\end{proof}

Recall that $\-{Lin}(H)$ is the line graph of $H$, that is, vertices in $\-{Lin}(H)$ are hyperedges
in $H$ and two hyperedges are adjacent if they share some vertex in
$H$. 
Let $\-{L}^2(H)$ be a graph whose vertices are hyperedges in $H$ and two hyperedges are adjacent in
$\-{L}^2(H)$ if their distance is at most $2$ in $\-{Lin}(H)$.  Any connected subgraph in
$\-{L}^2(H)$ contains a large $\{2,3\}$-tree in $\-{Lin}(H)$.

\begin{lemma}\label{lem:2-tree}
  Let $B$ be a set of hyperedges which induces a connected subgraph in $\-{L}^2(H)$, and $e^*\in B$
  be an arbitrary hyperedge.  There exists a $\{2,3\}$-tree $T\subseteq B$ such that $e^*\in T$ in
  $\-{Lin}(H)$ and $\abs{T}\ge\frac{\abs{B}}{k\Delta}$.
\end{lemma}
\begin{proof}
  We construct $T$ greedily starting from $T_0:=\{e^*\}$.  Given $T_i$, let
  $B\gets B\setminus\Gamma(T_i)$, and then let $T_{i+1}$ be $T_i$ plus the first hyperedge in $B$
  which has distance $\le 3$ from $T_i$.  If no such hyperedge exists, the process stops.

  We claim that when the process stops, all hyperedges in $B$ are removed.  If there is a nonempty
  subset $B'\subset B$ remaining, choose an arbitrary $e\in B'$.  Since $B$ is connected in
  $\-{L}^2(H)$, there is a shortest path $P\subset B$ from $e$ to some $e'\in T$ in $\-{L}^2(H)$.
  Assume that $P$ is $e\rightarrow \dots\rightarrow e_1\rightarrow e_2 \rightarrow e'$ (where $e_1$
  is possible to be $e$).  The minimality of $\abs{P}$ implies that $e_1,e_2\not\in T$.  If
  $\-{dist}_{\-{Lin}(H)}(T,e_2)=1$, then
  $\-{dist}_{\-{Lin}(H)}(T,e_1)\le 1 + \-{dist}_{\-{Lin}(H)}(e_1,e_2)\le 3$ and it contradicts the construction of $T$ as $e_1$ would be added to $T$.  Otherwise
  $\-{dist}_{\-{Lin}(H)}(T,e_2)=2$, and again it contradicts the construction of $T$ as $e_2$
  would be added to $T$.

  For the size of $T$, notice that in every step of the process, at most $k\Delta$ hyperedges are
  removed.  Hence $\abs{T}\ge\frac{\abs{B}}{k\Delta}$.
\end{proof}

We now define \emph{bad colourings}.  Let $e_0$ be the first edge in $\Gamma(v)$.  Recall that in
the coupling process we would attempt to colour at most $k_2$ vertices in an edge, where
$0<k_2<k_1$.  We will have another parameter $0<\beta<1$, which denotes the fraction of (partially)
monochromatic hyperedges\footnote{A hyperedge is
      (partially) monochromatic if every vertex in the hyperedge is either of the same colour or not
      coloured.} in a bad colouring.  All parameters will be set in Section
\ref{sec:parameters}.

\begin{definition}[bad colourings]\label{def:badcolouring}
  Let $\ell>0$ be an integer and $\beta>0$ be a constant.  A colouring $\sigma\in\+C_1$ is
  \emph{$\ell$-bad} if there exist a $\{2,3\}$-tree $T$ in $\-{Lin}(H)$ and vertices $V_{\-{col}}$ such that
  \begin{enumerate}
  \item $\abs{T} =\ell$ and $e_0\in T$;
  \item for every $e\in T$, $\abs{e\cap V_{\-{col}}} = k_2$;
  \item the partial colouring of $\sigma$ restricted to $V_{\-{col}}$
    makes at least $\beta\ell$ hyperedges in $T$ (partially) monochromatic.
  \end{enumerate}
  We say $\sigma\in\+C_1$ is $\ell$-\emph{good} if it is not $\ell$-bad.
\end{definition}

Note that since $T$ is a $\{2,3\}$-tree in $\-{Lin}(H)$ in Definition \ref{def:badcolouring}, all
hyperedges in $T$ are disjoint.

We show that the fraction of bad proper colourings among all proper colourings in $\+C_1$ is small.
This allows us to throw away bad colourings in the estimates later.

\begin{lemma}\label{lem:outside}
  Let $\Delta\ge 2$ and $0<k_2<k_1\le k$ all be integers.  Let $0<\beta<1$ be a constant.  Let
  $H(V,\+E)$ be a hypergraph with pinnings $\+P$, where the maximum degree is $\Delta$ and
  $k_1\le \abs{e}\le k$ for every $e\in\+E$.  If $q^{1-k_2} < \beta$,
  $q>\left( ek\Delta \right)^{\frac{1}{k_1-2}}$, and $q> \badthreshold$ where
  $C^{\beta(k_2-1)}\ge\frac{e^{\beta+3}k^3}{\beta^\beta}\cdot\binom{k}{k_2}$, then we have
  \[
    \frac{\abs{\set{\sigma\in\+C_1\cmid \text{$\sigma$ is
            $\ell$-bad}}}}{\abs{\+C_1}}\le\badfraction.
  \]
\end{lemma}
\begin{proof}
  Fix a $\{2,3\}$-tree $T=\set{e_1, e_2, \cdots, e_\ell}$ in $\-{Lin}(H)$ of size $\ell$ and
  $V_{\-{col}}$ such that for every $e\in T$, $\abs{e\cap V_{\-{col}}} =k_2$.  We say $\sigma$ is
  $\ell$-bad with respect to $T$ and $V_{\-{col}}$ if $\sigma$, $T$, and $V_{\-{col}}$ satisfy the
  requirments in Definition~\ref{def:badcolouring}.  Denote by $Z_{V_{\-{col}}}$ or simply $Z$ the
  number of (partially) monochromatic hyperedges by first drawing from $\mu_{\+C_1}$ and then
  revealing the colours of vertices in $V_{\-{col}}$.  We use Theorem \ref{thm:LLL-prob} to bound
  the probability that $Z\ge\beta\ell$.

  Indeed, $\mu_{\+C_1}$ can be viewed as the uniform distribution over proper colourings of an
  instance where $v$ is pinned to colour $c_1$.  In this instance, we have that
  $k_1-1\le\abs{e}\le k$ for every $e\in\+E$.  Hence, in the product distribution $\Pr{\text{$e$ is
      monochromatic}}\le q^{2-k_1} \le \frac{1}{ek\Delta}$ for every $e\in\+E$ by assumption.  We
  set $x(e)=\frac{1}{k\Delta}$ in Theorem~\ref{thm:LLL-prob} and verify \eqref{eqn:LLL}:
  \begin{align*}
    x(e)\prod_{e'\in\Gamma(e)}\tp{1-x(e')} \ge \frac{1}{k\Delta}\left( 1-\frac{1}{k\Delta} \right)^{k\Delta-1} 
    \ge \frac{1}{ek\Delta} \ge \Pr{\text{$e$ is monochromatic}}.
  \end{align*}

  In the product distribution (where all vertices are independent), for $e\in T$, the vertices in
  $e\cap V_{\-{col}}$ are monochromatic with probability $p^*\defeq q^{1-k_2} < \beta$.  Since $T$
  is a $\{2,3\}$-tree in $\-{Lin}(H)$, all edges are disjoint and these events are independent in
  the product distribution.  Hence, by a multiplicative Chernoff bound with mean $p^* \ell$ and
  $\gamma = \frac{\beta}{p^*}-1>0$,
  \begin{align*}
    \Pr{Z\ge \beta\ell}  = \Pr{Z \ge (1+\gamma)p^* \ell} 
    \le \left( \frac{e^{\gamma}}{(1+\gamma)^{1+\gamma}} \right)^{p^* \ell} \le \left( \frac{ep^*}{\beta} \right)^{\beta\ell}.
  \end{align*}
  For each edge $e\in T$, there are at most $k(\Delta-1)+1\le k\Delta-1$ edges that intersect with
  $e$ (including itself).  The random variable $Z$ thus depends on at most $(k\Delta-1)\ell$
  hyperedges in $\mu_{\+C_1}$.  By Theorem~\ref{thm:LLL-prob} with $x(e)=\frac{1}{k\Delta}$,
  \begin{align*}
    \mu_{\+C_1}(Z\ge \beta\ell) & \le \Pr{Z\ge \beta\ell}\cdot\left( 1-\frac{1}{k\Delta} \right)^{-(k\Delta-1)\ell}\\
                                & \le \left( \frac{ep^*}{\beta} \right)^{\beta\ell}\cdot e^{\ell} = \left( \frac{e^{1+1/\beta}p^*}{\beta} \right)^{\beta\ell}.
  \end{align*}

  To finish the argument, we still need to account for all $\{2,3\}$-trees and $V_{\-{col}}$ by an
  union bound.  Since the maximum degree in $\-{Lin}(H)$ is $k\Delta$, the total number of
  $\{2,3\}$-trees containing $e_0$ of size $\ell$, by Corollary \ref{cor:2-tree-number}, is at most
  $\frac{\tp{e(k\Delta)^3}^{\ell}}{2}$.  For a fixed $T$, since all edges in $T$ are disjoint, the
  number of possible $V_{\-{col}}$ is at most $\binom{k}{k_2}^\ell$.
  
  Putting everything together, we have that
  \begin{align*}
    \mPr[\sigma\sim\mu_{\+C_1}]{\text{$\sigma$ is $\ell$-bad}} 
    & \le \left( \frac{e^{1+1/\beta}p^*}{\beta} \right)^{\beta\ell} 
      \cdot \frac{\tp{e(k\Delta)^3}^{\ell}}{2}\cdot\binom{k}{k_2}^\ell\\
    & \le \left( \frac{e^{\beta+1}}{\beta^\beta}\cdot ek^3\cdot\binom{k}{k_2}\cdot q^{\beta-\beta k_2} \Delta^3\right)^{\ell}.
  \end{align*}
  By assumption, 
  \begin{align*}
    q^{\beta k_2-\beta} & \ge C^{\beta (k_2-1)} \Delta^{3}\ge
                          \frac{e^{\beta+2}}{\beta^\beta}\cdot ek^3\cdot \binom{k}{k_2}\cdot\Delta^{3}.
  \end{align*}
  Combining these two inequalities finishes the proof.
\end{proof}
  
Let $(x,y)\in\+T$ be a pair of partial colourings defined on $V_{\-{col}}$.
We are now going to prove some structural properties of $(x,y)$.  Say an edge $e\in\+E$ such that
$e\cap V_{\-{col}}\ne\varnothing$ is \emph{blocked} by $(x,y)$ if one of the following holds
\begin{enumerate}
\item $x(u)\ne y(u)$ for some $u\in e$.
\item $\abs{e\cap V_{\-{col}}}=k_2$ and $e$ is not satisfied by both $x$ and $y$.
\end{enumerate}
These two cases are called type 1 and type 2 errors respectively
in~\cite{Moi19}. Notice that all edges in $\Gamma(v)$ are always blocked, and in particular, $e_0$ is always blocked.

Let us denote the set of edges blocked by $(x,y)$ as $\+B_{x,y}$.  Then $\+B_{x,y}$ always contains
a large $\{2,3\}$-tree.
\begin{lemma}\label{lem:blocked_set}
  Let $(x,y)\in \+T$ be a pair of partial colourings in the coupling tree defined on $V_{\-{col}}$
  with corresponding $V_1$. Assume $\abs{V_{\-{col}}}=L$.  There exists a $\{2,3\}$-tree
  $T \subseteq \+B_{x,y}$ in $\-{Lin}(H)$ of size at least $\frac{L}{k^3\Delta^2}$ containing~$e_0$.
\end{lemma}
\begin{proof}
  We first claim that $\+B_{x,y}$ is connected in $\-{L}^2(H)$ by inducting on $L$.  Once an edge is
  blocked during Algorithm~\ref{alg:coupling}, it will remain blocked till the end.  If $u$ is the
  next vertex to be coloured in Algorithm~\ref{alg:coupling}, then $u$ must be adjacent to some
  vertex $u'\in V_1$, and $u'$ is in some edge $e$ blocked by the current $(x,y)$.  Therefore any
  newly blocked edge caused by colouring $u$ has distance at most $2$ to $e$.

  Since $e_0$ is always blocked, $e_0\in \+B_{x,y}$.  By Lemma~\ref{lem:2-tree}, there exists a
  $\{2,3\}$-tree $T\subseteq \+B_{x,y}$ in $\-{Lin}(H)$ such that
  $\abs{T}\ge \frac{\abs{\+B_{x,y}}}{k\Delta}$.  Next we claim that
  $\abs{\+B_{x,y}}\ge \frac{L}{k^2\Delta}$.  This is because that every vertex in $V_1$ belongs to
  some blocked edge.  Hence $\abs{V_1}\le k \abs{\+B_{x,y}}$.  By item (1) of Lemma
  \ref{lem:coupling}, $V_{\-{col}}\subseteq \Gamma_{\-{ver}}(V_1)$.  It implies that
  $L = \abs{V_{\-{col}}} \le \abs{\Gamma_{\-{ver}}(V_1)}\le k\Delta \abs{V_1}$.  Combining these
  facts yields the lemma.
\end{proof}

Recall that $\+T_L$ is the tree obtained from $\+T$ by truncating at depth $L$, and $\+L(\+T_L)$ is
its leaves.
Because of \textbf{Constraints 2}, for every proper colouring $\sigma\in\+C_1$, it holds that
\begin{equation}\label{eq:sumone}
  \sum_{\substack{(x,y)\in\+L(\+T_L):\;\sigma\models x}}p_{x,y}^x = 1.
\end{equation}
However, in \textbf{Constraints 1}, our linear program only contains constraints for those
$\ul{p}_{x,y}^x$ and $\ul{p}_{x,y}^y$ whose $V_{\-{col}}$ is of size strictly smaller than $L$.  The
next lemma shows that, for a $\ell$-good colouring $\sigma$, solving
$\ul{p}_{x,y}^x, \ul{p}_{x,y}^y$ provides a good approximation for the identity \eqref{eq:sumone}.

\newcommand{\Lvalue}{k^3\Delta^2\ell}

\begin{lemma}\label{lem:inside}
  Let $0<\beta<1$ be a constant.  Let $H=(V,\+E)$ be a hypergraph with pinnings $\+P$ and maximum
  degree $\Delta$ such that $\abs{e}\le k$ for all $e\in\+E$.  Let $\sigma\in\+C_1$ be $\ell$-good
  where $\ell$ is an integer.  If $\set{\widehat p_{x,y}^x}$ is a collection of values satisfying
  all our linear constraints, with $t^*=\tvalue$ in \textbf{Constraints 3} up to level $L=\Lvalue$,
  then it holds that
  \begin{align}\label{eqn:coupling-failure}
    \sum_{\substack{(x,y)\in\+L(\+T_L):\;\abs{V_{\-{col}}}<L\\
    \mbox{~\scriptsize and }\sigma\models x}}\widehat p_{x,y}^x \ge 1-\couplingfailure.
  \end{align}
\end{lemma}
\begin{proof}
  We construct a new coupling process similar to Algorithm \ref{alg:coupling}, and show the
  left-hand side of \eqref{eqn:coupling-failure} is the probability of an event defined by the new
  process.  We modify $\mathbb{S}$ in the following two ways: (1) condition on the final output
  being $\sigma$; (2) use probabilities induced by $\set{\widehat p_{x,y}^x}$ instead of
  $\set{\-p_{x,y}^x}$.  To be more specific, consider each step where one needs to extend $(x,y)$
  defined on $V_{\-{col}}$ to a new vertex $u$.  Call the new colourings $(x',y')$.  Since the
  output $\sigma$ is fixed, we simply reveal $x'(u)=\sigma(u)$.  In the original $\mathbb{S}$, the
  colour of $y'(u)$ is drawn according to an optimal coupling of $(x',y')$ on $u$.  Here, we set
  $y'(u)$ to colour $c$ with probability
  $\frac{\widehat p_{x^{u\gets\sigma(u)},y^{u\gets c}}^{x^{u\gets\sigma(u)}}}{\widehat p_{x,y}^x}$.
  This is well-defined since $\set{\widehat p_{x,y}^x}$ satisfies \textbf{Constraints 2}.  If this
  process reaches depth $L$, then it stops.

  The output of the new coupling defines a distribution over pairs of partial colourings $(x,y)$
  such that $\sigma \models x$ and we denote it by $\widehat\mu$.  We claim that
  \begin{equation}\label{eq:sumV1}
    \sum_{\substack{(x,y)\in\+L(\+T_L):\;\abs{V_{\-{col}}}=L\\\mbox{\scriptsize and } \sigma\models x}}\widehat p_{x,y}^x \le
    \sum_{\substack{\mbox{\scriptsize $\{2,3\}$-tree }T:\\\abs{T} = \ell,\;e_0\in T}}\mPr[(X,Y)\sim\widehat \mu]{T\subseteq \+B_{X,Y}}.
  \end{equation}

  Each summand on the left-hand side of \eqref{eq:sumV1} is the probability that our new coupling reaches some
  $(x,y)$ with $\abs{V_{\-{col}}}= L$.  Lemma~\ref{lem:blocked_set} implies that the set $\+B_{x,y}$
  of blocked edges contains a $\{2,3\}$-tree $T$ of size at least $\frac{L}{k^3\Delta^2}=\ell$.
  Thus the probability of reaching vertices of depth $L$ is upper bounded by the right-hand side of
  \eqref{eq:sumV1}.

  Fix a $\{2,3\}$-tree $T$ of size $\ell$.  Since $\sigma$ is $\ell$-good, whatever the choice of
  $V_{\-{col}}$ is, at least a $(1-\beta)$ fraction of hyperedges in $T$ must not be monochromatic
  on the $X$ side.  However, if $T\subseteq\+B_{X,Y}$, then at least $\floor{(1-\beta)\abs{T}}$
  hyperedges satisfy (1) $\sigma(v)\ne Y(v)$ for some $v\in e\cap V_{\-{col}}$, or (2)
  $\abs{e\cap V_{\-{col}}} = k_2$ and $\sigma|_{V_{\-{col}}}=X|_{V_{\-{col}}}$ satisfies $e$ but $Y$
  does not satisfy $e$.  It is clear that case (2) implies case (1), since if one partial colouring
  satisfies $e$ and another one does not, then they must differ at some $v\in e\cap V_{\-{col}}$.
  We use $T'=\set{e_1,e_2,\dots,e_{\abs{T'}}}$ to denote these hyperedges in $T$.  For each
  hyperedge in $T'$, there must be at least one vertex on which the (modified) coupling fails, which happens with
  probability at most $5/t^*$ due to \textbf{Constraints 3}.  Since $T$ is a $\{2,3\}$-tree in
  $\-{Lin}(H)$, all of these failed couplings are for distinct vertices and thus happen
  independently.  Hence, in this new coupling, the probability that every edge in $T'$ is blocked
  due to at least one failed vertex is at most
  $\left( \frac{5}{t^*} \right)^{\abs{T'}}\le \left( \frac{5}{t^*} \right)^{\floor{(1-\beta)\ell}}$.

  We still need to apply a union bound.  The number of $\{2,3\}$-trees of size $\ell$ in
  $\-{Lin}(H)$ and containing $e_0$ is, by Corollary \ref{cor:2-tree-number}, at most
  $\frac{\tp{e k^3\Delta^3}^{\ell}}{2}$.  Therefore the right-hand side of \eqref{eq:sumV1} is at
  most
  \begin{align}\label{eq:sumV1-bound}
    \sum_{\substack{\mbox{\scriptsize $\{2,3\}$-tree }T:\\\abs{T} = \ell,\;e_0\in T}}\mPr[(X,Y)\sim\widehat \mu]{T\subseteq \+B_{X,Y}}
    \le \left( \frac{5}{t^*} \right)^{\floor{(1-\beta)\ell}} \cdot \frac{\tp{e k^3\Delta^3}^{\ell}}{2}
    \le \couplingfailure, 
  \end{align}
  since we have chosen $t^*=\tvalue$ in \textbf{Constraints 3}.  The lemma follows by combining
  \eqref{eq:sumone}, \eqref{eq:sumV1}, and \eqref{eq:sumV1-bound}.
\end{proof}

Note that in Lemma~\ref{lem:inside} we do not explicitly require a lower bound of $q$ nor a lower
bound on the size of the edges.  However, these requirements are implicit since we have set $t^*$ to
be large in \textbf{Constraints 3}.

Lemma \ref{lem:outside} and Lemma \ref{lem:inside} also hold for any $\sigma\in\+C_2$.  Now we can
prove that any solution to the LP provides accurate estimates.

\begin{lemma}\label{lem:ratiobound}
  Assume the settings of Lemma~\ref{lem:outside} and Lemma~\ref{lem:inside}.  If the linear program
  up to level $L$ has a solution $\set{\widehat p_{x,y}^x,\widehat p_{x,y}^y}$ with guessed bounds
  $\set{\widehat \ratiolb,\widehat \ratioub}$, then it holds
  \[
    e^{-\gamma}\widehat\ratiolb \le \frac{\abs{\+C_1}}{\abs{\+C_2}} \le e^{\gamma} \widehat\ratioub,
  \]
  where $\gamma=\gammavalue$.
\end{lemma}
\begin{proof}
  Let $\ell=\ellvalue$.  Let
  \begin{align*}
    Z_1\defeq \sum_{\sigma\in\+C_1}
    \sum_{\substack{(x,y)\in\+L(\+T):\;\abs{V_{\-{col}}}<L\\\mbox{ \scriptsize and }\sigma\models x}}\widehat p_{x,y}^x.
  \end{align*}
  Exchange the order of summation:
  \begin{align*}
    Z_1 & = 
          \sum_{(x,y)\in\+L(\+T):\;\abs{V_{\-{col}}}<L}\sum_{\sigma\in\+C_1:\;\sigma\models x}\widehat p_{x,y}^x 
          = \sum_{(x,y)\in\+L(\+T):\;\abs{V_{\-{col}}}<L}\widehat p_{x,y}^x \cdot \abs{\+C_x}.
  \end{align*}
  A similar quantity $Z_2$ can be defined and bounded by replacing $\widehat p_{x,y}^x$ with
  $\widehat p_{x,y}^y$.  \textbf{Constraints~1} impose that for any $(x,y)\in\+L(\+T)$ such that
  $\abs{V_{\-{col}}}<L$,
  \begin{align*}
    \widehat\ratiolb\le\frac{\widehat{p}_{x,y}^x \cdot \abs{\+C_x}}{\widehat{p}_{x,y}^y \cdot \abs{\+C_y}}\le\widehat{\ratioub}.
  \end{align*}
  Hence,
  \begin{align}\label{eqn:ratio-Z}
    \widehat\ratiolb\le \frac{Z_1}{Z_2} \le \widehat\ratioub.
  \end{align}
  We will relate $\abs{\+C_1}$ with $Z_1$.  It is easy to see, by \eqref{eq:sumone}, that
  \begin{align}\label{eqn:lb-C1}
    \abs{\+C_1} & = \sum_{\sigma\in\+C_1}1 = \sum_{\sigma\in\+C_1} \sum_{(x,y)\in\+L(\+T_L):\;\sigma\models x}\widehat{p}_{x,y}^x \ge Z_1.
  \end{align}
  The lower bound is more complicated:
  \begin{align}
    \abs{\+C_1} & = \sum_{\sigma\in\+C_1}1 \le 
                  \tp{1-\badfraction}^{-1} \sum_{\substack{\sigma\in\+C_1:\\\text{\scriptsize $\sigma$ is $\ell$-good}}} 1 \notag\\
                &\le \tp{1-\badfraction}^{-1}\tp{1-\couplingfailure}^{-1}
                  \sum_{\substack{\sigma\in\+C_1:\\\text{\scriptsize $\sigma$ is $\ell$-good}}}
    \sum_{\substack{(x,y)\in\+L(\+T):\;\abs{V_{\-{col}}}<L\\\text{ \scriptsize and }\sigma\models x}} \widehat p_{x,y}^x \notag \\
                & \le e^{\gamma} \sum_{\sigma\in\+C_1}
                  \sum_{\substack{(x,y)\in\+L(\+T):\;\abs{V_{\-{col}}}<L\\\mbox{ \scriptsize and }\sigma\models x}}\widehat p_{x,y}^x = e^{\gamma}Z_1, \label{eqn:ub-C1}
  \end{align}
  where in the first line we use Lemma \ref{lem:outside} and in the second line we use Lemma
  \ref{lem:inside}.  Similar bounds hold with $\abs{\+C_2}$ and $Z_2$.  Combining
  \eqref{eqn:ratio-Z}, \eqref{eqn:lb-C1}, \eqref{eqn:ub-C1}, and their counterparts for
  $\abs{\+C_2}$ and $Z_2$, we have that
  \begin{align*}
    e^{-\gamma}\widehat\ratiolb &\le \frac{\abs{\+C_1}}{\abs{\+C_2}} \le e^{\gamma}\widehat\ratioub. \qedhere
  \end{align*}
\end{proof}


We then set up a binary search, to find $\ratiolb$ and $\ratioub$ that are close enough to the true
ratio.


We are now ready to prove the main theorem of this section.

\begin{proof}[Proof of Theorem~\ref{thm:marginal}]
  Take $L = k^3\Delta^2\ceil{\log\left(\frac{4}{\eps}\right)}$ so that
  $\gamma= \gammavalue\le \eps$.  We claim the true values of $\set{p_{x,y}^x,p_{x,y}^y}$ always
  satisfy our LP.  This is trivial for \textbf{Constraints 1} and \textbf{2}.  For
  \textbf{Constraints 3}, recall that $t^*=\tvalue > k$ and we only need to verify the conditions of
  Lemma~\ref{lem:LLL-prob} with $t=t^*$.  At any point of Algorithm \ref{alg:coupling}, the size of
  an edge is at least $k_1-k_2$.  Hence we set $k'=k_1-k_2$ in Lemma~\ref{lem:LLL-prob}.  By our
  assumption,
  \begin{align*}
    q >\couplingthreshold 
    \ge \left(5e\tp{e^2k^3}^{\frac{1}{1-\beta}}\right)^{\frac{1}{k'-1}} \cdot
    \Delta^{\frac{4-\beta}{(1-\beta)(k'-1)}}  
    = \left( et^*\Delta \right)^{\frac{1}{k'-1}}.
  \end{align*}  

  Fix the colour $c$.  It follows from Lemma~\ref{lem:ratiobound} that for every $c'\in[q]$,
  we can apply the binary search algorithm to obtain a value $p_{c'}$, which is an estimate of
  $\frac{\Pr[\sigma\sim\mu_{\+C}]{\sigma(v)=c'}}{\Pr[\sigma\sim\mu_{\+C}]{\sigma(v)=c}}$ satisfying
  \[
    e^{-\eps}\cdot p_{c'}\le
    \frac{\Pr[\sigma\sim\mu_{\+C}]{\sigma(v)=c'}}{\Pr[\sigma\sim\mu_{\+C}]{\sigma(v)=c}} \le
    e^{\eps}\cdot p_{c'}.
  \]
  We then use $\widehat{p}\defeq\tp{\sum_{c'\in [q]}p_{c'}}^{-1}$ to estimate
  $\Pr[\sigma\sim\mu_{\+C}]{\sigma(v)=c}$.

  For the running time, we treat $\Delta$, $k$, and $q$ as constants.  The size of the linear
  program in the WHILE loop is $\exp(O(L))$.  This is because the coupling tree $\+T$ is $q^2$-ary,
  and therefore it has at most $\exp(O(L))$ vertices up to depth $L$, and we have a pair of
  variables $\ul{p}_{x,y}^x$ and $\ul{p}_{x,y}^y$ for each vertex.  The number of variables and the
  number of constraints is at most $\exp(O(L))$.  Note that for each set of constraints in
  \textbf{Constraints 1}, we need to enumerate all the possible colourings in $V_1$ to compute
  $r_{x,y}$ for every leaf $(x,y)$.  This costs at most $\exp(O(L))$ time.  Hence it takes
  $\exp(O(L))$ time to construct an LP of size $\exp(O(L))$, which requires again $\exp(O(L))$ time
  to solve.  Note that with our choice of $L$, $\exp(O(L))=\poly\left( \frac{1}{\eps} \right)$.  For
  the WHILE loop, we use binary search to find $\ratiolb$ and $\ratioub$.  Thus the number of loops
  of the binary search is at most $\log_2\frac{2}{e^\eps}=\poly\left( \frac{1}{\eps} \right)$.
  Therefore, the total running time of our estimator is $\poly\left( \frac{1}{\eps} \right)$.
\end{proof}

\section{Approximate counting}\label{sec:count}

\newcommand{\acthreshold}{\ensuremath{\left(4(k-k_1^C)\Delta\right)^{\frac{1}{k-k_1^C-1}}}}
\newcommand{\couplingthresholdC}{\ensuremath{C\Delta^{\frac{5-\beta}{(1-\beta)(k_1^C-k_2-1)}}}}

Now we give our FPTAS for the number of proper $q$-colourings of a $k$-uniform hypergraph $H$ with
maximum degree $\Delta$.  The next lemma guarantees us a ``good'' proper colouring $\sigma$ so that
we can use the algorithm in Theorem \ref{thm:marginal} to compute the marginal probability of
$\sigma$.
 
\begin{lemma} \label{lem:order}
  Let $k_1^C$ be an integer such that $0<k_1^C< k-1$.  Let $q\ge\acthreshold$.  Let $v_1,\dots,v_n$
  be an arbitrary ordering of the vertices of a $k$-uniform hypergraph $H=(V,\+E)$.
  There exists a proper colouring $\sigma$ such that for every hyperedge $e\in\+E$, the partial
  colouring $\sigma$ restricted to the first $k-k_1^C$ vertices is not monochromatic.  Moreover,
  $\sigma$ can be found in deterministic polynomial time.
\end{lemma}

\begin{proof}
  Let $k'=k-k_1^C$.  Consider a new hypergraph $H'=(V,\+E')$ on the same vertex set $V$, but for
  every $e\in\+E$, we replace it with its first $k'$ vertices.  We set $x(e)=\frac{1}{k'\Delta}$ in
  Theorem~\ref{thm:LLL} and verify \eqref{eqn:LLL} for every $e\in\+E'$,
  \begin{align*}
    x(e)\prod_{e'\in\Gamma(e)}\tp{1-x(e')}
    & \ge \frac{1}{k'\Delta} \left( 1-\frac{1}{k'\Delta} \right)^{k'(\Delta-1)}
    & \ge \frac{1}{ek'\Delta} \ge q^{1-k'} 
    &\ge \Pr{\text{$e$ is monochromatic}}.
  \end{align*}
  Hence, Theorem~\ref{thm:LLL} implies that there exists a proper colouring $\sigma$ in $H'$, which
  satisfies the requirement of the lemma.

  In order to find $\sigma$, we have left a bit slack in our bound on $q$.  Thus the deterministic
  algorithm from \cite{MT10} applies.
\end{proof}

\begin{theorem}\label{thm:FPTAS}
  Assume the conditions of Theorem~\ref{thm:marginal} (on $q$, $\Delta$, $k$, $k_1$, $k_2$, and
  $\beta$) with $k_1=k_1^C$ hold, together with the conditions of Lemma~\ref{lem:order}.  There is
  an FPTAS for the number of proper $q$-colourings of a $k$-uniform hypergraph $H=(V,\+E)$ with
  maximum degree~$\Delta$.
\end{theorem}

\begin{proof}
  Let $n=\abs{V}$.  Choose an arbitrary ordering of the vertices $v_1,\dots,v_n$ of $V$.
  Lemma~\ref{lem:order} implies that we can find a proper colouring $\sigma$ so that any hyperedge
  is properly coloured by the first $k-k_1^C$ of its vertices.  Let $Z=\abs{\+C}$ be the number of
  proper colourings of $H$.  For every $\eps>0$, we will deterministically compute a number
  $\widehat{Z}$ in time polynomial in $n$ and $1/\eps$ such that
  $e^{-\eps}\widehat{Z} \le Z \le e^\eps\widehat{Z}$.

  As before, let $\mu_{\+C}$ be uniform over $\+C$, the set of all proper colourings of $H$.  We
  will actually estimate $\mu_{\+C}(\sigma) = \frac{1}{Z}$.  To this end, we create a sequence of
  hypergraphs $\{H_i\}$ with pinnings $\{\+P_i\}$ inductively.  Let $H_1=H$ and $\+P_1$ be empty.
  Given $H_i=(V_i,\+E_i)$ and $\+P_i$, we find the next vertex $u_i$ under the ordering that are
  contained in at least one hyperedge of $H_i$.  We pin the colour of $u_i$ to be $\sigma(u_i)$.
  This induces a pinning $\+P_{i+1}$ on all hyperedges in $\+E_i$.  Then, $H_{i+1}$ is obtained by
  removing $u_i$ from $V_i$ and removing all hyperedges that are properly coloured under $\+P_{i+1}$
  from $\+E_i$.  We also truncate the pinning $\+P_{i+1}$ accordingly.  If for some $n'\le n$,
  $\+E_{n'}$ is empty, then this process terminates.  Notice that the construction above yields a
  subset of vertices $u_1,\dots,u_{n'}$ where $n'\le n$.  Their ordering is consistent with the
  given ordering.

  We claim that for any $i\in[n']$, for any $e\in\+E_i$, it satisfies that
  $k_1^C \le \abs{e} \le k$.  This is because an edge $e$ shrinks in size in the process when
  vertices are pinned according to $\sigma$.  However, Lemma \ref{lem:order} guarantees that the
  edge $e$ will be removed in the process above before $k-k_1^C$ vertices are coloured.  Therefore,
  together with our assumptions, Theorem \ref{thm:marginal} applies with $k_1=k_1^C$.
   
  Let $p_i$ be the marginal probability of colour $\sigma(u_i)$ at $u_i$ in $H_i$ with pinning
  $\+P_i$.  Let $p_i=\frac{1}{q}$ for all $i\ge n'$.  It is easy to see that
  $Z^{-1} = \mu_{\+C}(\sigma) = \prod_{i=1}^{n}p_i$.
  Thus we can obtain our desired estimate $\widehat{Z}$ by approximating each $p_i$ within
  $e^{\pm \frac{\eps}{n}}$.  To this end, we appeal to Theorem \ref{thm:marginal} with
  $\eps'=\frac{\eps}{n}$.
\end{proof}

\section{Sampling} \label{sec:sampling} Finally we give the algorithm to sample proper colourings
almost uniformly.  As usual, let $H(V,\+E)$ be a $k$-uniform hypergraph with maximum degree
$\Delta$, $q$ be the number of colours, and $\+C$ be the set of proper colourings.  Let $n=\abs{V}$.
Algorithm~\ref{algo:sampling} samples a colouring in $\+C$ within total variation distance $\eps$
from $\mu_{\+C}$.  Similar to the coupling process in Section~\ref{sec:coupling}, we assume that
there is an arbitrary fixed ordering of all vertices and hyperedges.  There is a parameter
$0<k_1^S<k-1$ in Algorithm \ref{algo:sampling}, which will be set in Section~\ref{sec:parameters}.

\newcommand{\sizethreshold}{\ensuremath{k^2\Delta\log\left( \frac{2n\Delta}{\eps} \right)}}
\newcommand{\twotreesize}{\ensuremath{\frac{L}{k^2\Delta}}}

\begin{algorithm}[htbp]
  \caption{An almost uniform sampler for proper colourings}
  \begin{algorithmic}[1]
    \State \textbf{Input:} \emph{A $k$-uniforom hypergraph $H(V,\+E)$ with maximum degree $\Delta$
      and $0<\eps<1$}%
    \State \textbf{Output:} \emph{A colouring in $\+C$}%
    \State Let $X$ be the partial colouring that $X(v)=-$ for every $v\in V$ initially;%
    \While{$\+E$ is nonempty}%
    \State Choose the first uncoloured $v\in V$ such that every $e\in\Gamma(v)$ contains $>k_1^S$
    uncoloured vertex;%
    \If {no such vertex $v$ exists} \State \textbf{break} \EndIf \State Apply the algorithm in
    Theorem~\ref{thm:marginal} to compute the marginal distribution on $v$ with precision
    $\frac{\eps}{2n}$, and extend $X$ with the colour on $v$ according to the
    distribution; \label{algoL:marginal}%
    \State Remove from $\+E$ all hyperedges that are now satisfied.  \EndWhile \State $S\gets$
    uncoloured vertices in $V$;%
    \State Let $H_S=(S,\+E_S)$ where $\+E_S\defeq \set{e\cap S\cmid e\in \+E}$;%
    \If {$H_S$ contains a connected component with size at least
      $\sizethreshold$} \label{algoL:fail}%
    \State \Return an arbitary $x\in\+C$ \Else \State \Return a uniformly random proper colouring
    consistent with $X$ by enumerating all proper colourings of $H_S$.  \EndIf
  \end{algorithmic}
  \label{algo:sampling}
\end{algorithm}

We first assume that at Line~\ref{algoL:marginal}, the oracle call to Theorem~\ref{thm:marginal} is
always within the correct range.  This simplification allows us to identify a threshold involving
the parameter $k_1^S$ to guarantee small connected components, which will be put together with the
conditions of Theorem~\ref{thm:marginal} later.

\newcommand{\ccthreshold}{\ensuremath{C\Delta^{\frac{3}{k-k_1^S-1}}}}
\newcommand{\couplingthresholdS}{\ensuremath{C\Delta^{\frac{4-\beta}{(1-\beta)(k_1^S-k_2-1)}}}}

\begin{lemma}\label{lem:fail}
  Assume the oracle call to Theorem~\ref{thm:marginal} at Line~\ref{algoL:marginal} is within the
  desired range.  If $q>\left( ek\Delta \right)^{\frac{1}{k_1^S-1}}$ and $q>\ccthreshold$ where
  $ C^{(k-k_1^S)-1} > e^7k^3$, the condition in line~\ref{algoL:fail} of Algorithm
  \ref{algo:sampling} holds with probability at most $\eps/2$.
\end{lemma}
\begin{proof}
  The proof idea is to show the existence of a large components in $H_S$ implies the existence of a
  large $\{2,3\}$-tree in $\-{Lin}(H)$ whose vertices are edges that are not satisfied but $k-k_1^S$
  of their vertices are already coloured. Then we show the probability of the latter event is small.

  Now assume that the sampler ends the WHILE loop with a partial colouring $X$ and $H_S$.  We say an
  edge $e\in \+E$ is \emph{bad} if $X$ does not satisfy $e$ and $\abs{e\cap S}=k_1^S$, namely $e$ is
  partially monochromatic under $X$ but $k-k_1^S$ vertices have been coloured.  Also, say a vertex
  $v\in S$ is \emph{blocked} by an edge $e\in \+E$ if $v\in e$ and $e$ is bad.

  Fix an arbitrary hyperedge $e_0$ that is bad, and $e_0$ is contained in a connected component of
  size at least $L$ in $H_S$.  We denote the set of vertices of this component by $U$ and its
  induced hypergraph $H_U$.  It is clear that every vertex in $S$ is blocked by some bad edge.  Let
  $\+F$ be the set of all bad edges incident to $U$.  Then $e_0\in\+F$.  Since every vertex in $U$
  is blocked by some edge in $\+F$ and every edge in $\+F$ contains at most $k$ vertices,
  $\abs{\+F}\ge\frac{L}{k}$.
  
  We claim that $\+F$ is connected in $\-L^2(H)$.  The reason is the following.  For any two edges,
  say $e_1,e_2\in \+F$, since $H_U$ is connected, there exists a path in $H_U$ connecting $e_1$ and
  $e_2$.  Every vertex along this path must be blocked by some edge in $\+F$.  Each adjacent pair of
  vertices along this path corresponds to a pair of edges in $\+F$ that have distance at most $2$ in
  $\-{Lin}(H)$.
  
  Lemma~\ref{lem:2-tree} implies that $\+F$ contains a $\{2,3\}$-tree of size at least
  $\ell=\twotreesize$ containing $e_0$.  Fix such a $\{2,3\}$-tree $T=\set{e_1,\dots,e_{\abs{T}}}$.
  Let $\widehat\mu$ be the distribution of our sampler at the end of the WHILE loop.  It holds that
  \[
    \mPr[X\sim\widehat\mu]{\mbox{\scriptsize every $e_i\in T$ is
        bad}}=\prod_{i=1}^{\abs{T}}\mPr[X\sim\widehat\mu]{\mbox{$e_i$ is bad}\mid
      \bigwedge_{j<i}\mbox{$e_j$ is bad}}.
  \]
  Since $e_i\cap e_j=\varnothing$ for every $i\ne j$ and Theorem \ref{thm:marginal} guarantees our
  estimated marginals are within $e^{\eps/2n}$, for every $1\le i\le \abs{T}$, we can apply Lemma
  \ref{lem:LLL-prob} with $k'=k_1^S$ and $t=k$,
  \begin{align*}
    \mPr[X\sim\widehat\mu]{\mbox{$e_i$ is bad}\mid
      \bigwedge_{j<i}\mbox{$e_j$ is bad}}  
    \le q\cdot q^{-(k-k_1^S)}\cdot (1+8/t)^{k/2} \cdot e^{\frac{\eps (k-k_1^S)}{2n}}    \le e^5\cdot q^{1-(k-k_1^S)}.
  \end{align*}
  Applying Lemma~\ref{lem:LLL-prob} requires that $q>\left( ek\Delta \right)^{\frac{1}{k_1^S-1}}$.
  By Corollary \ref{cor:2-tree-number}, the number of $\{2,3\}$-trees of size $\ell$ in $\-{Lin}(H)$
  containing $e_0$ in $\+F$ is at most $\frac{\tp{ek^3\Delta^3}^\ell}{2}$.  Then by the union bound,
  the probability that $H_S$ contains a component with size at least $L$ is at most
  \begin{align}\label{eqn:failprob}
    n\Delta\tp{ek^3\Delta^3}^\ell \left(e^5\cdot q^{1-(k-k_1^S)}\right)^{\ell},
  \end{align}
  where the term $\abs{n\Delta}\ge\abs{\+E}$ accounts for the choice of $e_0$.  By assumption,
  \begin{align*}
    q^{(k-k_1^S)-1} > C^{(k-k_1^S)-1} \Delta^{3} > e^7k^3\Delta^{3}.
  \end{align*}
  As $L = \sizethreshold$ and $\ell= \twotreesize$, $e^{-\ell} \le \frac{\eps}{2n\Delta}$.  Hence,
  by \eqref{eqn:failprob} the probability in Line \ref{algoL:fail} is at most
  \begin{align*}
    n\Delta\tp{ek^3\Delta^3}^\ell \left(e^5\cdot q^{1-(k-k_1^S)}\right)^{\ell} 
    & \le n\Delta \cdot e^{-\ell} \le \frac{\eps}{2}. \qedhere
  \end{align*}
\end{proof}

Now we are ready to give the sampling algorithm.

\begin{theorem}\label{thm:sampling}
  Assume the conditions of Theorem~\ref{thm:marginal} (on $q$, $\Delta$, $k$, $k_1$, $k_2$, and
  $\beta$) with $k_1=k_1^S$ hold, together with the conditions of Lemma~\ref{lem:fail}.  For any
  $k$-uniform hypergraph $H=(V,\+E)$ with maximum degree $\Delta$ and $\eps>0$,
  Algorithm~\ref{algo:sampling} outputs a proper colouring whose distribution is within $\eps$ total
  variation distance to the uniform distribution, and the running time is
  $\-{poly}(n,\frac{1}{\eps})$ where $n=\abs{V}$.
\end{theorem}
\begin{proof}
  First we check that the condition of Theorem \ref{thm:marginal} is met with $k_1=k_1^S$, when it
  is called in Algorithm~\ref{algo:sampling} at Line~\ref{algoL:marginal}.  This is because whenever
  we colour a vertex, we make sure that all hyperedges have at least $k_1^S$ uncoloured vertices
  afterwards.  Hence we apply Theorem \ref{thm:marginal} with the pinnings $\+P$ induced by the
  partial colouring $X$ so far.
  
  We use $\widehat\mu(\cdot)$ to denote the distribution of the final output of
  Algorithm~\ref{algo:sampling}.  Recall thet $\mu_{\+C}$ is the uniform distribution over $\+C$.
  We shall bound the total variation distance $\-{dist}_{TV}(\mu_{\+C},\widehat\mu)$.  To this end,
  we introduce two intermediate distributions: Let $\mu_1(\cdot)$ be the distribution obtained from
  the output of Algorithm~\ref{algo:sampling} but ignoring the condition on line~\ref{algoL:fail} in
  Algorithm \ref{algo:sampling}.  Namely, it never checks the size of connected components in $H_S$
  and proceed to enumerate all the proper colourings on $S$ in any case.  This is unrealistic since
  doing so would require exponential time.  We also define another distribution $\mu_2(\cdot)$,
  which is the same as $\mu_1(\cdot)$ except at line~\ref{algoL:marginal}, it uses the true marginal
  instead of the estimate by calling Theorem~\ref{thm:marginal}.

  Denote by $B$ the event that the condition on line~\ref{algoL:fail} holds.  Let $p_{\-{fail}}$ be
  the probability of event $B$.  By Lemma \ref{lem:fail}, $p_{\-{fail}}\le\eps/2$.

  First note that $\mu_2=\mu_{\+C}$. Consider the distribution of the partial colouring obtained
  immediately after the WHILE loop, i.e., the partial colouring $X$. One can apply induction similar
  to the proof of Lemma~\ref{lem:coupling} to show that it follows a pre-Gibbs distribution.
  Therefore, conditioned on $X$, sampling a uniform proper colouring of the remaining vertices
  results in a uniform proper colouring.

  We then bound $\-{dist}_{TV}(\mu_1,\mu_2)$. For a particular partial colouring $x$, we use $E_x$
  to denote the event that the sampler produces $x$ at the end of the WHILE loop, namely $X=x$.  It
  holds that
  \begin{align*}
    \-{dist}_{TV}(\mu_1,\mu_2)
    &=\frac{1}{2}\sum_{\sigma\in\+C}\abs{\mPr[Z\sim\mu_1]{Z=\sigma}-\mPr[Z\sim\mu_2]{Z=\sigma}}\\
    &=\frac{1}{2}
      \sum_{\sigma\in\+C}
      \Bigg|\sum_{\substack{x:\;\sigma\models x}}\bigg(\mPr[Z\sim\mu_1]{Z=\sigma\mid E_x}
      \cdot\mPr[Z\sim\mu_1]{E_x} 
     -\mPr[Z\sim\mu_2]{Z=\sigma\mid
      E_x}\cdot\mPr[Z\sim\mu_2]{E_x}\bigg)\Bigg|,
  \end{align*}
  where $x$ runs over partial colourings.

  The partial colouring $x$ may never appear at the end of the WHILE loop in Algorithm
  \ref{algo:sampling}.  In this case,
  \begin{align*}
    \Pr[Z\sim\mu_1]{E_x}=\Pr[Z\sim\mu_2]{E_x}=0.
  \end{align*}
  Otherwise $x$ can be the partial colouring at the end of the WHILE loop.  Since the enumeration
  steps are identical and correct in both $\mu_1$ and $\mu_2$ conditioned on $E_x$, we have that
  \begin{align*}
    \Pr[Z\sim \mu_1]{Z=\sigma\mid E_x}=\Pr[Z\sim \mu_2]{Z=\sigma\mid E_x}=\frac{\*1_{\sigma\models x}}{\abs{\+C_x}},
  \end{align*}
  where $\+C_x$ is again the set of proper colourings consistent with the partial colouring $x$.
  
  It implies that
  \begin{equation}\label{eq:gap12}
    \-{dist}_{TV}(\mu_1,\mu_2)=\frac{1}{2}\sum_{\sigma\in\+C}\abs{\sum_{\substack{x:\;\sigma\models x}}\frac{1}{\abs{\+C_x}}\tp{\mPr[Z\sim\mu_1]{E_x}-\mPr[Z\sim\mu_2]{E_x}}}.
  \end{equation}
  Fix a partial colouring $x$ defined on $V_{\-{col}}\subseteq V$ that is a possible output of the
  WHILE loop.  We note that the order of visiting $V_{\-{col}}$ is determined by the random choices
  of $x$.  Say this order is $v_1,\dots,v_s$.  Let
  \begin{align*}
    p_i\defeq\Pr[Z\in\mu_C]{Z(v_i)=x(v_i)\mid \bigwedge_{1\le j<i}Z(v_j)=x(v_j)}.
  \end{align*}
  Hence
  \[
    \mPr[Z\sim\mu_1]{E_x}-\mPr[Z\sim\mu_2]{E_x}=\prod_{i=1}^s \widehat{p}_i-\prod_{i=1}^s p_i,
  \]
  where $\widehat{p}_i$ is our estimate of $p_i$ using Theorem~\ref{thm:marginal} with error
  $\frac{\eps}{2n}$.  Theorem~\ref{thm:marginal} implies that
  \[
    e^{-\frac{\eps}{2n}} \widehat{p}_i \le p_i \le e^{\frac{\eps}{2n}} \widehat{p}_i.
  \]
  Therefore, we have
  \begin{equation}\label{eq:gap12-aux}
    \abs{\mPr[Z\sim\mu_1]{E_x}-\mPr[Z\sim\mu_2]{E_x}} \le \eps \mPr[Z\sim\mu_2]{E_x}.
  \end{equation}
  Plugging \eqref{eq:gap12-aux} into \eqref{eq:gap12}, we obtain
  \begin{align*}
    \-{dist}_{TV}(\mu_1,\mu_2)\le
    \frac{1}{2}\sum_{\sigma\in\+C}\abs{\sum_{\substack{x:\;\sigma\models
    x}}\frac{\eps}{\abs{\+C_x}}\mPr[Z\sim\mu_2]{E_x}}= \frac{\eps}{2}\sum_{\sigma\in\+C}\mu_2(\sigma)=\frac{\eps}{2}.
  \end{align*}

  Finally we bound $\-{dist}_{TV}(\widehat\mu,\mu_1)$.  Since the behaviours of $\widehat\mu$ and
  $\mu_1$ are identical if $B$ does not happen, we have that
  $\mPr[Z\sim\widehat\mu]{Z=\sigma\mid
    \overline{B}}${}$ = \mPr[Z\sim\mu_1]{Z=\sigma\mid \overline{B}}$.  It implies that
  \begin{align*}
    \-{dist}_{TV}(\widehat\mu,\mu_1) 
    &=\frac{1}{2}\sum_{\sigma\in\+C}\abs{\mPr[Z\sim\widehat\mu]{Z=\sigma}-\mPr[Z\sim\mu_1]{Z=\sigma}}\\
    &=\frac{1}{2}\sum_{\sigma\in\+C}\left|\mPr[Z\sim\widehat\mu]{Z=\sigma \land B} +
      \mPr[Z\sim\widehat\mu]{Z=\sigma\mid \overline{B}}\cdot (1-p_{\-{fail}})\right.\\
    &\quad\quad\quad\left.-\mPr[Z\sim\mu_1]{Z=\sigma \land B} -
      \mPr[Z\sim\mu_1]{Z=\sigma\mid \overline{B}}\cdot (1-p_{\-{fail}})\right|\\
    &=\frac{1}{2}\sum_{\sigma\in\+C}\abs{\mPr[Z\sim\widehat\mu]{Z=\sigma \land B}
      -\mPr[Z\sim\mu_1]{Z=\sigma \land B}}\\
    &\le \frac{1}{2}\sum_{\sigma\in\+C}\tp{\mPr[Z\sim\widehat\mu]{Z=\sigma \land
      B}+\mPr[Z\sim\mu_1]{Z=\sigma \land B}}\\
    &\le p_{\-{fail}}.
  \end{align*}

  Combining the above and Lemma~\ref{lem:fail}, we obtain
  \begin{align*}
    \-{dist}_{TV}(\widehat\mu,\mu_{\+C})&\le
                                          \-{dist}_{TV}(\widehat\mu,\mu_1)+\-{dist}_{TV}(\mu_1,\mu_2)+\-{dist}_{TV}(\mu_2,\mu_{\+C})
    \\
                                        &\le
                                          p_{\-{fail}}+\frac{\eps}{2}\le \eps.
  \end{align*}

  It remains to bound the running time of the sampler. The sampler calls subroutines to estimate
  marginal at most $n$ times and each time the subroutine costs
  $\-{poly}(n,\frac{1}{\eps})$. Finally, upon the condition on line~\ref{algoL:fail} does not hold,
  the sampler enumerates proper colourings on connected components of size
  $O(\log\tp{\frac{n}{\eps}})$. Therefore, the total running time is $\-{poly}(n,\frac{1}{\eps})$.
\end{proof}

The distribution $\mu_1$ has a small multiplicative error comparing to the uniform distribution
$\mu_{\+C}$.  We remark that there are standard algorithms to turn such a distribution into an exact
sampler, dating back to \cite{Bac88,JVV86}.  However, since we cannot completely avoid event $B$, we
can only bound the error in the final distribution $\widehat{\mu}$ in terms of total variation
distance.

\section{Settling all parameters}\label{sec:parameters}

We have defined the following parameters throughout the paper:
\begin{itemize}
\item $k_1^{C}$: the number of vertices in a hyperedge that are \emph{not} fixed in approximate
  counting, Theorem~\ref{thm:FPTAS};
\item $k_1^{S}$: the number of vertices in a hyperedge that are \emph{not} fixed in sampling,
  Theorem \ref{thm:sampling};
\item $k_2$: the number of vertices in a hyperedge Algorithm~\ref{alg:coupling} would attempt to
  couple;
\item $\beta$: the fraction of hyperedges that are monochromatic in Definition
  \ref{def:badcolouring}.
\end{itemize}

We want our bound for approximate counting to have the form $C\Delta^\frac{A_1}{k - B_1}$.  By
Theorem~\ref{thm:FPTAS}, we want to make sure that, for any $k>0$, subject to $0<k_2<k_1^C<k-1$, and
$0<\beta<1$,
\begin{align*}
  \frac{A_{1}}{k - B_{1}} & \ge \frac{3}{\beta(k_2-1)};\\
  \frac{A_{1}}{k - B_{1}} & \ge \frac{4-\beta}{(1-\beta)(k_1^C-k_2-1)};\\
  \frac{A_{1}}{k - B_{1}} & \ge \frac{1}{k-k_1^C-1}.
\end{align*}
We assume $k_1^C$ and $k_2$ are proportional to $k$.  Minimizing $A_1$ yields the following
solutions:
$A_1 = 14,B_1=14,k_1^C=\floor{\frac{13k}{14}}, k_2= \floor{\frac{3k}{7}},\beta=\frac{1}{2}$.
Plugging these values into Theorem~\ref{thm:FPTAS}, we want to satisfy the following constraints:
\begin{align*}
  k-k_1^C-2 & \ge 0, &
                       C & \ge \left(5e\tp{e^2k^3}^{\frac{1}{1-\beta}}\right)^{\frac{1}{k_1^C-k_2-1}},\\
  q^{k_2-1} & > \frac{1}{\beta}, &
                                   C & \ge \left(\frac{e^{\beta+3}k^3}{\beta^\beta}\cdot\binom{k}{k_2}\right)^{\frac{1}{\beta(k_2-1)}},\\
  q & >\left( ek\Delta \right)^{\frac{1}{k_1^C-2}}, & 
                                                      C & \ge 4(k-k_1^C)^{\frac{1}{k-k_1^C-1}}.
\end{align*}
One can verify that $k\ge \kone$ and $C\ge \Cone$ suffice.  This yields Theorem~\ref{thm:main}.

Similarly, we want our bound for sampling to have the form $C\Delta^\frac{A_{2}}{k - B_{2}}$.  By
Theorem~\ref{thm:sampling}, we want to make sure that, for any $k>0$, subject to $0<k_2<k_1^S<k-1$
and $0<\beta<1$,
\begin{align*}
  \frac{A_{2}}{k - B_{2}} & \ge \frac{3}{\beta(k_2-1)};\\
  \frac{A_{2}}{k - B_{2}} & \ge \frac{4-\beta}{(1-\beta)(k_1^S-k_2 -1)};\\
  \frac{A_{2}}{k - B_{2}} & \ge \frac{3}{k-k_1^S-1}.  
\end{align*}
Similarly to the approximate counting case, minimizing $A_2$ yields the following solutions:
$A_2 = 16,B_2=\frac{16}{3},k_1^S=\floor{\frac{13k}{16}},k_2 =
\floor{\frac{3k}{8}},\beta=\frac{1}{2}$.  Plugging these values into Theorem~\ref{thm:sampling}, we
want to satisfy the following constraints:
\begin{align*}
  k-k_1^S-2 & \ge 0, &
                       C & \ge \left(5e\tp{e^2k^3}^{\frac{1}{1-\beta}}\right)^{\frac{1}{k_1^S-k_2-1}},\\
  q^{k_2-1} & > \frac{1}{\beta}, &
                                   C & \ge \left(\frac{e^{\beta+3}k^3}{\beta^\beta}\cdot\binom{k}{k_2}\right)^{\frac{1}{\beta(k_2-1)}},\\
  q & >\left( ek\Delta \right)^{\frac{1}{k_1^S-2}}, & 
                                                      C & > \tp{e^7k^3}^{\frac{1}{(k-k_1^S)-1}}.
\end{align*}
One can verify that $k\ge \kone$ and $C\ge \Ctwo$ suffice.  This yields Theorem~\ref{thm:main2}.  We
note that these constraints also hold for $k\ge 6$ and $C\ge 1.2\times 10^{11}$.

\section{Concluding remarks} \label{sec:conclusion}

In this paper we give approximate counting and sampling algorithms for hypergraph colourings, when
the parameters are in the local lemma regime.  One important open question is how to get an optimal
constant in the exponent of $\Delta$ in Theorem~\ref{thm:main} and~\ref{thm:main2}.  This constant
comes from three places: to bound the number of ``bad colourings'' (Lemma~\ref{lem:outside}), to
bound the error (in the LP) incurred by ``good colourings'' (Lemma~\ref{lem:inside}), and finally to
leave some slack for either counting (Theorem~\ref{thm:FPTAS}) or sampling
(Theorem~\ref{thm:sampling}).  It seems to us that the last slack is difficult to reduce, and a
tighter result, if possible, would come from improvements on the first two parts, although our
analysis has been pushed to the limit.


Another future direction is to generalize this approach for general constraint satisfaction problems
(CSP), or equivalently, the general setup of the (variable version) local lemma.  Our analysis
relies on some crucial property of hypergraph colourings, that all constraints can be satisfied by
partial assignments, ideally with appropriate probabilities.  To be more specific, suppose a
constraint $C$ contains $k$ variables.  We require a property that, when a subset of $k'$ variables
are randomly assigned, the probability that $C$ is still not satisfied is roughly $ c^{-k'}$ for
some constant $c >1$.  This property does not necessarily hold in general, even for symmetric
constraints.  One such example is when the variables take values from $[q]$, and the constraint is
satisfied unless the sum of all its variables is $0$ modulo $q$.  We can take $q$ to be large so
that the strong local lemma conditions hold, and yet this constraint cannot be satisfied by any
subset of variables.  In particular, it is problematic to bound our definition of ``bad colourings''
(Definition~\ref{def:badcolouring}) when constraints cannot be satisfied by partial assignments.
New ideas are required to handle more general settings.

Upon closer look, the success of our approach does not truly rely on that the
system is in the local lemma regime. What is essential is that the coupling tree can be truncated at
a suitable depth without incurring big error. This turns out to be a special form of the \emph{spatial
mixing} property. In the settings of this paper, a strong form of the local lemma condition guarantees
that the coupling process succeeds with sufficiently high probability at each step and therefore 
establishes the desired property. A consequence is that we can use \emph{local} linear constraints to 
certify the coupling. It remains unclear whether a global correlation decay argument would suffice as well.



\section*{Acknowledgements}

We thank anonymous referees for many improvements,
and in particular, for pointing out the connection to spatial mixing properties.

\bibliographystyle{alpha} \bibliography{bib}

\end{document}